%% file: main.tex
\documentclass[twoside,leqno]{article}
\usepackage{hyperref}
\hypersetup{breaklinks, urlcolor=blue, colorlinks, citecolor=green!50!black, linkcolor=blue}
\usepackage[letterpaper, left=1in, right=1in, top=0.9in, bottom=0.9in]{geometry}
\usepackage[utf8]{inputenc}
\usepackage[american]{babel}
\usepackage[normalem]{ulem}
\usepackage{amsmath, amssymb, cases, amsthm}
\usepackage{thmtools}
\usepackage[shortlabels]{enumitem}
\usepackage{mdframed}
\usepackage{bbm}
\usepackage{bm}
\usepackage{microtype}
\usepackage{xcolor}
\usepackage{makecell}
\usepackage{mathtools}
\usepackage{float}
\usepackage{modletters}
\usepackage{comment}
\usepackage{multirow}
\usepackage{natbib}
\usepackage[capitalize,noabbrev]{cleveref}
\usepackage{graphics}
\usepackage{fancyhdr}
\usepackage{footnotehyper}
\makesavenoteenv{tabular}

\usepackage{pifont}

\declaretheorem[numberwithin=section,refname={Theorem,Theorems},Refname={Theorem,Theorems}]{theorem}

\declaretheorem[numberlike=theorem]{lemma}

\declaretheorem[numberlike=theorem]{corollary}
\declaretheorem[numberlike=theorem,style=definition]{definition}
\declaretheorem[numberlike=theorem]{claim}
\declaretheorem[numberlike=theorem,style=remark]{remark}

\declaretheorem[numberlike=theorem,refname={Conjecture,Conjectures},Refname={Conjecture,Conjectures},name={Conjecture}]{conjecture}

\def\final{1}  %
\def\iflong{\iffalse}
\ifnum\final=0  %

\newcommand{\yonggang}[1]{{\color{blue}[{\tiny Yonggang: \bf #1}]\marginpar{*}}}
\newcommand{\danupon}[1]{{\color{red}[{\tiny Danupon: \bf #1}]\marginpar{\color{red}*}}}

\newcommand{\sagnik}[1]{{\color{green!50!black}[{\tiny Sagnik: \bf #1}]\marginpar{\color{green!50!black}*}}}
\newcommand{\todo}[1]{{\color{red}[{\tiny TODO: \bf #1}]\marginpar{\color{red}*}}}
\newcommand{\yuval}[1]{{\bf \color{red!50!black} YUVAL: #1}}
\newcommand{\jan}[1]{{\bf \color{green!50!black} Jan: #1}}
\newcommand{\blikstad}[1]{\textup{\color{magenta} [\textbf{Joakim}: #1]}}
\newcommand{\tawei}[1]{{\color{blue} [{\bf Ta-Wei:} #1]}}
\newcommand{\TODO}[1]{{\color{blue!50!black} [{\bf Todo:} #1]}}
\else %
\newcommand{\yonggang}[1]{}
\newcommand{\danupon}[1]{}
\newcommand{\sagnik}[1]{}
\newcommand{\todo}[1]{}
\newcommand{\yuval}[1]{}
\newcommand{\jan}[1]{}
\newcommand{\blikstad}[1]{}
\newcommand{\tawei}[1]{}
\newcommand{\TODO}[1]{}
\fi  %
\usepackage[ruled,vlined,linesnumbered]{algorithm2e}

\newcommand{\eps}{\varepsilon}

\newcommand{\polylog}{\mathrm{polylog}}
\newcommand{\poly}{\mathrm{poly}}

\newcommand{\set}[2][ ]{\{#2 \ifthenelse{\equal{#1}{ }}{ }{~|~#1}\}}

\newcommand{\N}{\mathbb{N}}

\newcommand{\tOh}{\widetilde{O}}

\usepackage{caption2}
\usepackage{empheq}

\newcommand{\The}[1]{\tilde{\Theta}\left(#1\right)}
\newcommand{\pset}{\mathcal{P}}
\newcommand{\aset}{\mathcal{A}}

\renewcommand{\L}{\mathcal{L}}
\newcommand{\I}{\mathcal{I}}
\renewcommand{\ss}{shortcut}
\newcommand{\ssss}[2]{$(#2,#1)$-shortcut}

\newcommand{\labcov}{{\sf LabelCover}}

\newcommand{\mlab}{multilabeling}

\newcommand{\os}{{\sf MinShC}}
\newcommand{\oss}[2]{$(#2,#1)$-\os{}}
\newcommand{\tc}{TC spanner}
\newcommand{\TC}[2]{$(#2,#1)$-{\sf MinTC}}
\newcommand{\TCs}[2]{$(#2,#1)$-\tc{}}
\newcommand{\conj}{PGC}
\newcommand{\distt}[2]{\text{dist}_{#1}\left(#2\right)}

\newcommand{\AB}{\Delta}
\newcommand{\gadget}[2]{$#2$-{\sf MinStShC}$\mid_{#1}$}
\newcommand{\ga}{{\sf MinStShC}}
\newcommand{\lan}{{c_{lrs}}}
\newcommand{\bufs}{{s'}}

\newcommand{\minre}{LabelCover}

\newcommand{\epsl}{\epsilon_{L}}
\renewcommand{\a}[1]{a^{(#1)}}
\renewcommand{\aa}[2]{a^{(#1)}_{#2}}
\newcommand{\aaa}[3]{\alpha^{(#1)}_{#2}[#3]}
\renewcommand{\b}[1]{b^{(#1)}}
\newcommand{\bb}[2]{b^{(#1)}_{#2}}
\newcommand{\bbb}[3]{\beta^{(#1)}_{#2}[#3]}
\newcommand{\dia}{\mathcal{\rho}}
\newcommand{\cs}{c_M}
\newcommand{\idx}{I}
\newcommand{\Idx}[1]{\idx{}[#1]}

\SetKwComment{Comment}{/* }{ */}

\newcommand{\apxS}{\alpha_S}
\newcommand{\apxD}{\alpha_D}

\Crefname{algocf}{Algorithm}{Algorithms}

\usepackage[most]{tcolorbox}

\usepackage{xparse}
\usepackage{lipsum}

\makeatletter
\newcommand\footnoteref[1]{\protected@xdef\@thefnmark{\ref{#1}}\@footnotemark}
\makeatother

\renewcommand{\paragraph}[1]{\medskip\noindent{\bf #1}\xspace}

\title{Shortcuts and Transitive-Closure Spanners Approximation}

\author{
Parinya Chalermsook \thanks{University of Sheffield, \texttt{chalermsook@gmail.com}}\and
Yonggang Jiang\thanks{MPI-INF, Germany, \texttt{yjiang@mpi-inf.mpg.de}} \and 
Sagnik Mukhopadhyay\thanks{University of Birmingham \texttt{s.mukhopadhyay@bham.ac.uk}} \and
Danupon Nanongkai\thanks{MPI-INF, Germany, \texttt{danupon@gmail.com}}
}
\date{}

\begin{document}
	
	\begin{titlepage}
		\maketitle 
         
        \pagenumbering{roman}

\input{abstract}

		\setcounter{tocdepth}{3}
		\newpage
		\tableofcontents
		\newpage
        
	\end{titlepage}
	
	\newpage
	\pagenumbering{arabic}

\input{introduction}

\input{preliminaries}

\input{overview}

\input{lowerbounds}

\newpage

\appendix

\input{upper}

\input{open.tex}
\bibliography{Bibliography}

\end{document}

%% file: abstract.tex
\begin{abstract}

We study polynomial-time approximation algorithms for two closely-related problems, namely computing shortcuts and transitive-closure spanners (TC spanners). 
For a directed unweighted graph $G=(V, E)$ and an integer $d$, a set of edges $E'\subseteq V\times V$ is called a $d$-TC spanner of $G$ if the graph $H:=(V, E')$ has  (i) the same transitive-closure as $G$ and (ii) diameter at most $d.$ The set $E''\subseteq V\times V$ is a $d$-shortcut of $G$ if $E\cup E''$ is a $d$-TC spanner of $G$. Our focus is on the following $(\alpha_D, \alpha_S)$-approximation algorithm: given a directed graph $G$ and integers $d$ and $s$ such that $G$ admits a $d$-shortcut (respectively $d$-TC spanner) of size $s$, find a $(d\alpha_D)$-shortcut (resp. $(d\alpha_D)$-TC spanner) with $s\alpha_S$ edges, for as small $\alpha_S$ and $\alpha_D$ as possible. These problems are important special cases of graph sparsification and arise naturally in the context of reachability problems across computational models. 

As our main result, we show that, under the Projection Game Conjecture (PGC), there exists a small constant $\epsilon>0$, such that no polynomial-time $(n^{\epsilon},n^{\epsilon})$-approximation algorithm exists for finding $d$-shortcuts as well as $d$-TC spanners of size $s$. Previously, super-constant lower bounds were known only for $d$-TC spanners with constant $d$ and ${\alpha_D}=1$ [Bhattacharyya, Grigorescu, Jung, Raskhodnikova, Woodruff 2009]. Similar lower bounds for super-constant $d$ were previously known only for a more general case of directed spanners [Elkin, Peleg 2000]. No hardness of approximation result was known for shortcuts prior to our result.

As a side contribution, we complement the above with an upper bound of the form $(n^{\gamma_D}, n^{\gamma_S})$-approximation which holds for $3\gamma_D + 2\gamma_S > 1$ (e.g., $(n^{1/5+o(1)}, n^{1/5+o(1)})$-approximation). The previous best approximation factor is obtained via a naive combination of known techniques from [Berman, Bhattacharyya, Makarychev, Raskhodnikova, Yaroslavtsev 2011] and [Kogan, Parter 2022], and can provide $(n^{\gamma_D}, n^{\gamma_S})$-approximation under  the condition $3\gamma_D + \gamma_S >1$; in particular, for a fixed value $\gamma_D$, our improvement is nearly quadratic.  
\end{abstract}

%% file: introduction.tex
\section{Introduction} \label{sec:intro}

For a directed unweighted graph $G=(V,E)$, a \emph{$d$-shortcut} of $G$ is a set $E' \subseteq V \times V$ such that, for every pair of vertices $u,v \in V$,
\begin{enumerate}
  \item[(i)] if $u$ can reach $v$ in $E \cup E'$, then $u$ can reach $v$ in $G$, and
  \item[(ii)] if $u$ can reach $v$ in $G$, then $u$ can reach $v$ in $E \cup E'$ via a path of at most $d$ edges.
\end{enumerate}
We call $d$ the \emph{depth} of~$E'$.

Shortcuts are a crucial tool for parallel reachability\footnote{The reachability problem asks, given a directed graph $G=(V,E)$ and two vertices $s,t \in V$, whether $s$ can reach $t$, i.e., whether there is a directed path from $s$ to $t$.}: given a $d$-shortcut $E'$, it suffices to run BFS on $E \cup E'$ up to depth~$d$ to determine reachability, which requires linear work and depth~$d$. Indeed, all known work-efficient\footnote{A work-efficient parallel reachability algorithm has almost linear total work (i.e., the number of unit operations).} parallel algorithms rely on the efficient construction of shortcuts~\cite{Fineman18,JinST20,CaoFR20}.

Typically, a work-efficient parallel reachability algorithm is expected to construct a $d$-shortcut of almost-linear size\footnote{Throughout the paper, we use \emph{almost linear} to refer to $m^{1+o(1)}$ and \emph{near linear} to refer to $m \cdot \polylog(m)$, where $m$ is the number of edges. One might notice that most of the parallel reachability works rely on shortcuts of size $n \polylog(n)$ instead of $m \polylog(n)$. However, we clarify that constructing a shortcut of size $m \polylog(n)$ is a trivial upper bound for $n \polylog(n)$ and suffices for work efficiency.}, while keeping $d$ as small as possible. Can we hope for $d$ to be polylogarithmic, as conjectured by Thorup~\cite{Thorup92}? Unfortunately, several works have exhibited pathological worst-case graph families for which every almost-linear-size shortcut must have depth at least $n^{1/5}$~\cite{Hesse03,HuangP21,BodwinH23,WilliamsXX24}.

Real-world graphs, however, are seldom such worst-case instances; they usually exhibit additional structure that allows better shortcuts. For example, planar graphs admit shortcuts of near-linear size and polylogarithmic depth~\cite{Thorup95}. Motivated by this, this paper explores the following research question:
\begin{center}
  \itshape Is there an algorithm that can (approximately) compute the best shortcut for a given directed graph?
\end{center}
If shortcuts are indeed the only viable approach to work-efficient parallel reachability---as all existing algorithms suggest---then such an algorithm would lead to an algorithm that is beyond-worst-case optimal, in the sense that it is optimal among all shortcut-based algorithms.

\paragraph{Formalizing approximation.}
To “approximately” compute the best shortcut, suppose that $G$ admits a shortcut of almost-linear size $s = m^{1+o(1)}$ and depth~$d$. We naturally would like an algorithm that outputs a shortcut of almost-linear size as well, with depth $d \cdot \apxD$, where the slackness parameter $\apxD$ is small, ideally $n^{o(1)}$. Formally, the bicriteria approximation task is (see \Cref{sec:prelim} for details):
\begin{quote}
  \emph{Given a directed graph $G$ and integers $d$ and $s$ such that $G$ admits a $d$-shortcut of size $s$, find a $(d\apxD)$-shortcut with $s\apxS$ edges, where $\apxD,\apxS = n^{o(1)}$.}
\end{quote}
We call such algorithms \emph{$(\apxD,\apxS)$-approximation algorithms}. We are mainly interested in the case where $s \ge n$ and $d \gg \log n$, since it suffices in most applications.

\paragraph{Connection to TC-spanners.}
A set $E' \subseteq V \times V$ is a \emph{$d$-TC spanner}\footnote{TC stands for “transitive closure,” which is the edge set connecting every pair of vertices $(u,v)$ such that $u$ can reach $v$.} of $G$ if
\begin{enumerate}
  \item[(i)] $E'$ is a subgraph of the transitive closure of $G$, and
  \item[(ii)] for all reachable pairs $u,v$ in $G$, there is a path of at most $d$ edges from $u$ to $v$ in $E'$.
\end{enumerate}
TC-spanners have been studied much more extensively in the literature; see~\cite{Raskhodnikova10} for a comprehensive survey.

We define $(\apxD,\apxS)$-approximation algorithms for finding a TC-spanner in a manner analogous to shortcuts. One can readily see the similarity between the definitions of shortcuts and TC-spanners. Indeed, we can prove that they are essentially equally hard to approximate in the settings we are interested in (see \cref{lem:redShtoTC,lem:redTCtoSh} for details).\danupon{Is it possible to explain in a few lines what “equivalent” roughly means?}

\paragraph{Existing works (upper bounds).}
Studies of shortcut upper bounds have mainly focused on the worst-case size~\cite{KoganP22,KoganP-icalp22,KoganP23}. This line of work shows that there is a polynomial-time algorithm that computes an $n^{1/3}$-shortcut of size $\tilde{O}(n)$\footnote{We use $\tilde{O}(\cdot)$ to hide polylogarithmic factors.}. This trivially implies a $(n^{1/3},\polylog(n))$-approximation algorithm (since we assume $s \ge n$). Extending this argument to exploit the trade-off provided by~\cite{KoganP22} yields a $(n^{\gamma_D}, n^{\gamma_S})$-approximation factor for all $3\gamma_D + \gamma_S > 1$ (e.g.\ an $(n^{1/4+o(1)}, n^{1/4+o(1)})$-approximation factor).

For TC-spanners, \cite{BermanBMRY11} obtains a single-criterion $(1, O(\sqrt{n}\log n))$ approximation. One of our contributions is to show that, by combining \cite{BermanBMRY11} and \cite{KoganP22}, we can obtain a bicriteria algorithm for both TC-spanners and shortcuts that improves on \cite{KoganP22} (e.g.\ an $(n^{1/5+o(1)}, n^{1/5+o(1)})$-approximation factor), although the approximation factors remain polynomial. See \Cref{sec:upperbound} for more details.

In short, no existing algorithm achieves $\apxS,\apxD = n^{o(1)}$ simultaneously, leaving open whether a subpolynomial bicriteria approximation is possible.

\paragraph{Existing works (lower bounds).}
While strong inapproximability results are known for general spanners~\cite{ElkinP07}\footnote{By slightly adapting their lower bound graphs, a bi-criteria lower bound can also be proved for spanner (not TC-spanner). Thus, our lower bound is a breakthrough only for TC-spanner, not general spanner.}, analogous bounds for TC-spanners (and thus shortcuts) are scarce. Bhattacharyya et al.~\cite{BhattacharyyaGJRW12} showed a conditional lower bound stating that, for \emph{constant} $d$, no polynomial-time $(1, 2^{\log^{1-\varepsilon} n})$ approximation exists\footnote{They use a weaker conditional assumption so it results in $2^{\log^{1-\varepsilon} n}$. If they use the same assuption as us (PGC), then they will get $n^\eps$ as weill - so here we can view the lower bound as $(1,n^{\eps})$, which does not change our discussion.}. However, this result does not address the important regime where $d$ is superconstant, nor does it preclude a $(1, n^{o(1)})$ single-criterion approximation. Consequently, whether both $\apxD$ and $\apxS$ can be made subpolynomial remains widely open.

\paragraph{Our results.}
We prove a strong conditional lower bound that completely rules out the possibility of simultaneously subpolynomial $\apxD$ and $\apxS$.

\begin{theorem}[Conditional lower bound; informal]\label{thm:intro:lowerbound}
  Under the \emph{Projection Games Conjecture (PGC)}, there exists a constant $\varepsilon>0$ such that no polynomial-time $(n^{\varepsilon}, n^{\varepsilon})$-approximation algorithm can find $d$-shortcuts (resp.\ $d$-TC spanners) of size~$s$.\qedhere
\end{theorem}

Our result relies on a weak variant of the \emph{Projection Games Conjecture} (PGC)~\cite{Moshkovitz15}, an important conjecture in the field of hardness of approximation. PGC underlies the inapproximability of many central problems whose hardness factors are derived from Label Cover, including directed multicut and sparsest cut~\cite{chuzhoy2009polynomial}, set cover~\cite{Moshkovitz15}, and spanners~\cite{DinitzKR16}.
PGC remains consistent with all known algorithmic techniques and has emerged as a central barrier to improving approximation algorithms for many fundamental problems (although not mentioned explicitly by this name).\footnote{Some of these results predate the formal statement of PGC by Moshkovitz, and therefore do not reference it explicitly. For example, Chuzhoy and Khanna~\cite{chuzhoy2009polynomial} attribute their stronger hardness result to a specific PCP, which is in fact equivalent to the weak variant of PGC that we use in this paper.}

Our lower bound overcomes the long-standing difficulty of proving strong lower bounds for TC-spanners, as suggested by~\cite{Ras10,BhattacharyyaGJRW12}. This could shed light on lower-bound techniques for shortcut- and TC-spanner-related problems. Beyond that, some interesting consequences follow from our original motivation for studying these problems.

Our results suggest that it might be impossible to efficiently find an almost-best shortcut for a given graph, thus eliminating the hope of having an “almost optimal” shortcut-based parallel reachability algorithm for every input graph structure. This indicates that we should shift our focus to the following two directions in pursuit of beyond-worst-case parallel algorithms for reachability:

\begin{description}
    \item[Efficient algorithms for specific graph families.] Although we cannot hope for a universal algorithm that finds the best shortcut for every input graph, we may hope for an algorithm that finds the best shortcut for specific graph families, as has been done for planar graphs~\cite{Thorup95}. Other graph classes, such as low-treewidth graphs, may also admit efficient shortcut algorithms and remain largely unexplored.

    \item[Improving the polynomial factor for approximation.] Although our lower bound suggests that an almost-optimal approximation algorithm does not exist, it does not rule out the possibility of reasonably good bicriteria algorithms, such as a $(n^{0.1}, n^{0.1})$-approximation, which could be useful in many applications. We emphasize that although the $\varepsilon$ in our lower bound may be very small, we do not aim to optimize the polynomial factor; our focus is solely to rule out subpolynomial hardness. Closing the gap between the upper and lower bounds remains a difficult but interesting open question.
\end{description}

\paragraph{Related Work.}
In a recent work \cite{DinitzKN25}, the authors studied \textbf{single-criteria} settings (in our terminology, the case where $\alpha_D = 1$), providing both upper and lower bounds for \textbf{shortcuts and hopsets}, but not addressing transitive-closure spanners. They established several interesting upper bounds for single-criteria shortcut and hopset approximations, which are incomparable to our upper bounds (presented in the appendix), as our focus is primarily on the bi-criteria setting.
Regarding lower bounds, their results for single-criteria shortcuts and hopsets are adapted from the bounds in \cite{ElkinP07}, and correspond to our warm-up section (\Cref{subsec:warmup}). The main challenge and novelty of our work lie in extending these lower bounds to the bi-criteria setting (\Cref{subsec:bicriterialower}), and further, to transitive-closure spanners (\Cref{subsec:largebicriteria}).

%% file: preliminaries.tex
\section{Preliminaries} \label{sec:prelim}

\paragraph{Graph terminology.} Given a directed graph (digraph) $G=(V,E)$, a directed edge, denoted by an ordered pair of vertices $(u,v)\in E$, is called an in-edge for $v$ and an out-edge for $u$. The vertices $u$ and $v$ are called the endpoints of the edge $(u,v)$. The out-degree of a node $u\in V$ denotes the number of out-edges $(u,v)\in E$, and the in-degree denotes the number of in-edges $(v,u)\in E$. The maximum in-degree and out-degree of $G$ are the highest in-degree and out-degree values among all nodes, respectively.

A path $p$ of length $\ell$ is a vertex sequence $(v_0,v_1,...,v_\ell)$ where $(v_i,v_{i+1})\in E$ for any $0\le i<\ell$. The following phrases are equivalent and used interchangeably: $p$ is a path from $v_0$ to $v_\ell$, $v_0$ has a path to $v_\ell$, $v_0$ can reach $v_\ell$, $v_\ell$ is reachable from $v_0$, and $(v_0,v_\ell)$ is a reachable (ordered) pair. The distance between a pair of vertices $(v_0, v_\ell)$ is defined as the minimum length of paths from $v_0$ to $v_\ell$, denoted by $\distt{G}{v_0,v_\ell}$. The diameter of graph $G$ is defined as the maximum distance between any two reachable vertex pairs.

Given a digraph $G$, we use $G^T = (V, E^T )$ to represent the transitive closure of $G$. In other words, $(u, v) \in E^T $ if and only if vertex $u$ can reach vertex $v$ in $G$. For an edge $(u, v) \in E$, vertices $u$ and $v$ are referred to as the endpoints of the edge $(u, v)$. A subgraph of $G$ is a graph $G' = (V', E')$, where $E' \subseteq E$ and $V'$ contain all vertices that are endpoints of edges in $E'$. We will slightly stretch the terminology and use $E'$ to also denote the subgraph $G'$. For a given set of vertices $V' \subseteq V$, the subgraph induced by $V'$ is denoted by $G[V']$. 

\subsection{Shortcut and \tc{}}\label{subsec:shortcut}

We first introduce the definition of \tc{}.

\begin{definition}[\tc{}]
For a digraph $G=(V,E)$, any subset of edges $E' \subseteq E^T$ that has the same transitive closure of $G$ is called a \emph{\tc{}} of $G$. If $E'$ has diameter at most $d$ and size at most $s$, we call $E'$ as \TCs{s}{d}.

\end{definition}

The following problem is about \tc{} approximation.

\begin{definition}[\TC{\apxS}{\apxD}]
Given a directed graph $G$ and integers $d$ and $s$, such that $G$ admits a \TCs{s}{d}, the goal is to find a \TCs{\apxS s}{\apxD d}.
\end{definition}

\begin{definition}[Shortcut]
For a digraph $G=(V,E)$, any subset of edges $E' \subseteq E^T \setminus E$ is called a \emph{\ss{}} of $G$. For $s, d \in \N$, we say $E'$ is a \emph{\ssss{s}{d}} of $G$ if $|E'| \le s$ and the diameter of $E\cup E'$ is at most $d$.
\end{definition}

The following problem is about \ss{} approximation.

\begin{definition}[\oss{\apxS}{\apxD}]
Given a directed graph $G$ and integers $d$ and $s$, such that $G$ admits a \ssss{s}{d}, the goal is to find an \ssss{\apxS s}{\apxD d}.
\end{definition}

When we write $\apxS,\apxD$ as a function of $n$, the variable $n$ corresponds to the number of nodes present in graph $G$.

\subsection{Reductions between Shortcut and \tc{}}\label{subsec:reductions}

In this section, we prove an approximate equivalence between the two problems. 

\begin{lemma}\label{lem:redTCtoSh}
	If there is a polynomial time algorithm solving \oss{\apxS}{\apxD}, then there is a polynomial time algorithm solving \TC{\apxS+1}{\apxD}.  
\end{lemma}

\begin{proof}
	Given a graph $G$ that admits a \TCs{s}{d}, we will show how to find a \TCs{(\apxS+1) s}{\apxD d}.
	
	Firstly, we will try to find a minimum subgraph of $G$ with the same transitive closure as $G$. This is well-studied in \cite{AhoGU72} and can be done in polynomial time. To be precise, according to Theorem 2~\cite{AhoGU72}, we can assume $G$ is a DAG. According to Theorem 1 \cite{AhoGU72}, a unique subgraph $H$ exists such that deleting any edge of $H$ will make the transitive closure of $G^T$ and $H$ unequal. Therefore, we can find $H$ in polynomial time: start with $G^T$, and repeatedly delete edges if deleting an edge will not change the transitive closure until no edge can be deleted. Since $G$ admits a \TCs{s}{d}, and a \tc{} must have the same transitive closure as $G$, we have $|E(H)|\le s$. %
 
 After finding $H$, we use the \ssss{\apxS}{\apxD} algorithm with input $H$ and $s,d$ to find a \ssss{\apxS s}{\apxD d} $E'$. $H$ admits a \ssss{s}{d} because the \TCs{s}{d} of $G$ is also a \ssss{s}{d} of $H$. We claim that the graph $(V(H), E'\cup E(H))$ is a \TCs{(\apxS+1) s}{\apxD d}. According to the definition of shortcut, $(V(H), E'\cup E(H))$ has a diameter at most $\apxD d$. The size is at most $\apxS s+|E(H)|\le (\apxS+1) s$ since $|E(H)|\le s$.
\end{proof}

\begin{lemma}\label{lem:redShtoTC}
	If there is a polynomial time algorithm solving \TC{\apxS}{\apxD}, then there is a polynomial time algorithm solving \oss{2\apxS}{\apxD} with input $s$ restricted to $s\ge m$.
\end{lemma}
\begin{proof}
	Given a graph $G$ with $m$ edges and $s\ge m,d$ such that $G$ admits a \ssss{s}{d}, we will show how to find a \ssss{2\apxS s}{\apxD d}. 
 
        We first prove that $G$ admits a \TCs{2s}{d}. Let the \ssss{s}{d} of $G$ be $E'$, then $E'\cup E$ has size at most $2s$ since $s\ge m$, and has diameter $d$. Thus, $E'\cup E$ is a \TCs{2s}{d} of $G$. 
        
        We apply the \TC{\apxS}{\apxD} algorithm on $G$ with inputs $2s,d$ to get a \TCs{2\apxS s}{\apxD d}, which is also a \ssss{2\apxS s}{\apxD d}. 
\end{proof}

\subsection{Label cover and PGC.}\label{subsec:labcov}

In this section, we define the label cover problem and state the projection games conjecture formally. Notice that this is not new, but already appears in several previous works as standard tools \cite{ElkinP07,Kortsarz01}.

\begin{definition}[\labcov{}]\label{def:labcov}
	The \labcov{} problem takes as input an instance $\I=(A, B,E,\L,$ $(\pi_e)_{e\in E})$ described as follows:
	\begin{itemize}
		\item $(A \cup B, E)$ is a bipartite regular graph with partitions $|A|=|B|$, 
		\item a label set $\L$ and a non-empty acceptable label pair associated with each $e=(u,v)\in E$ denoted by $\pi_e\subseteq \L\times \L$.
	\end{itemize}
	
	A \emph{labeling} $\psi: V \to \cL$ is a function that gives a label to each vertex $\psi(v)\in\L$. $\psi$ is said to cover an edge $e=(u,v)\in E$ if $(\psi(u),\psi(v))\in\pi_e$. The goal of the problem is to find a labeling that covers the most number of edges.
	
\end{definition}

\begin{conjecture}[The Weak Variant of PGC]\label{con:pgc}
	There exists a universal constant $\epsilon >0$ such that, given a \labcov{} instance on input of size $N$, it is hard to distinguish between the following two cases: 
	\begin{itemize}
		\item (Completeness:) There is a labeling that covers every edge. 
		
		\item (Soundness:) Any labeling covers at most $N^{-\epsilon}$ fraction of the edges. 
	\end{itemize}
\end{conjecture} 

We will use a slightly stronger hardness result where the soundness case is allowed to assign multiple labels. The hardness of this variant is a simple implication of the PGC.  A \mlab{} $\psi:V \to 2^\cL$ gives a set of labels  $\psi(v)\subseteq\L$ to a vertex $v$. Such a multilabeling  is said to cover an edge $e=(u,v)\in E$ if there exists $a\in\psi(u),b\in\psi(v)$ such that $(a,b)\in\pi_e$. 
The \textit{cost} of $\psi$ is denoted by $\sum_{u \in A \cup B} |\psi(u)|$. Notice that a valid labeling is a multi-labeling of cost $|A|+|B|$.

\begin{lemma}\label{lem:pgc}
	Assuming PGC, 
	there exists a sufficiently small constant $\epsl{} >0$ such that, given a \labcov{} instance of size $N$, there is no randomized polynomial time algorithm that distinguishes between the following two cases: 
	\begin{itemize}
		\item (Completeness:) There is a labeling that covers every edge. 
		
		\item (Soundness:) Any multilabeling of cost at most $N^{\epsl{}}(|A|+|B|)$ covers at most $N^{-\epsl{}}$ fraction of edges. 
	\end{itemize}
\end{lemma}
\begin{proof}
    Suppose there is an algorithm that can distinguish the two cases described in~\cref{lem:pgc}, we will show that it also distinguishes the two cases in~\cref{con:pgc}. Completeness in~\cref{con:pgc} trivially implies completeness in~\cref{lem:pgc}, we only need to show that it also holds for soundness. In other words, we want to show that if there exists a \mlab{} $\phi$ of cost at most $N^{\epsl}(|A|+|B|)$ that covers more than $N^{-\epsl}$ fraction of edges for sufficiently small constant $\epsl$, then there exists a labeling $\phi'$ that covers more than $N^{-\epsilon}$ fraction of edges.

    To construct $\phi'$, we uniformly at random sample $1$ label from $\phi(v)$ for every $v\in A\cup B$ as the label $\phi'(v)$. For each edge $(u,v)$ covered by $\phi$, the probability that $(u,v)$ is covered by $\phi'$ is at least $1/(|\phi(u)|\cdot |\phi(v)|)$. Let $P$ contain all edges $(u,v)$ covered by $\phi$ such that $|\phi(u)|<N^{3\epsl}$ and $|\phi(v)|<N^{3\epsl}$. Since $\sum_{u\in A\cup B}|\phi(u)|\le N^{\epsl}(|A|+|B|)$, the number of nodes $u\in A\cup B$ that $|\phi(u)|\ge N^{3\epsl}$ is bounded by $(|A|+|B|)/N^{2\epsl}$. Remember that the bipartite graph is regular, so the number of edges $(u,v)$ with $|\phi(u)|\ge N^{3\epsl}$ or $|\phi(v)|\ge N^{3\epsl}$ is at most 
    \[\frac{2|E|}{|A|+|B|}\cdot \frac{|A|+|B|}{N^{2\epsl}}=\frac{2|E|}{N^{2\epsl}}\]

    Since the number of edges covered by $\phi$ is at least $N^{-\epsl}|E|$, we have
    \[|P|\ge \frac{|E|}{N^{\epsl}}-\frac{2|E|}{N^{2\epsl}}\ge \frac{|E|}{N^{2\epsl}}\]
    The expectation of the number of covered edges by $\phi'$ is at least
    \[\sum_{(u,v)\in P}\frac{1}{|\phi(u)|\cdot |\phi(v)|}\ge \frac{|E|}{N^{8\epsl}}\]

    By setting $\epsl<\epsilon/8$, the expectation is at least $\frac{|E|}{N^{\epsilon}}$, which means there must exist a labeling $\phi'$ that covers more than $N^{-\epsilon}$ fraction of edges.
\end{proof}

\subsection{Main Results}\label{subsec:mainresults}

We restate the main results in the introduction in a formal way.

\begin{theorem}[Lower bound]\label{thm:main}
		Under \conj{}, there exists a small constant $\epsilon>0$, such that no polynomial-time algorithm exists for solving \oss{n^{\epsilon}}{n^{\epsilon}} as well as \TC{n^{\epsilon}}{n^{\epsilon}}, even if we restricted the input $s=\Omega(m^{1+\epsilon}),d=\Omega(n^{\epsilon})$.
\end{theorem}

Recall that the Las Vegas algorithm is one that either generates a correct output or, with probability at most $1/2$, outputs $\bot$ to indicate its failure. 

\begin{restatable}[Upper bound]{theorem}{UpperBoundThm}\label{thm:upperbound}

  There is a polynomial time Las Vegas algorithm solving \oss{\apxS}{\apxD} whenever the inputs $G,s,d$ satisfy the following conditions.
  \begin{enumerate}
      \item $s\ge |V(G)|$  and %
      \item $\apxS\ge\frac{Cn\log^2 n}{\sqrt{sd^2\apxD^3}}$ for sufficiently large constant $C$.
  \end{enumerate}

Moreover, if we further require $\apxS\ge 2$, then there is a polynomial time  Las Vegas algorithm solving \TC{\apxS}{\apxD}.  
\end{restatable}

Notice that the upper bound result is stated here in a slightly different form, as this form will be more convenient to derive formally. It is easy to show that Theorem~\ref{thm:upperbound} implies $(n^{\gamma_D}, n^{\gamma_S})$-approximation for $3\gamma_D + 2\gamma_S >1$ (using the fact that $s \geq n$ and $d \geq 1$.)   

%% file: overview.tex
\section{Overview of Techniques}
\label{sec:overview}

\danupon{To do: Define things not defined in the intro. If there are too many things, consider moving this after the prelim.}

Here we give a brief overview of our approach and techniques to prove \Cref{thm:main}. For convenience, we refer to the problem of approximating shortcut of a given directed graph as \os{}. Indeed, our goal is to prove that \os{} is $(n^{\epsilon}, n^{\epsilon})$-hard to approximate even when assuming that $s=\Omega(m)$ and diameter $d=n^{\epsilon}$ for some small constant $\epsilon$.\footnote{Notice that since we are proving a lower bound, restricting the class of inputs only strengthens the result. This restriction is necessary in order for our bound to also imply a lower bound for transitive-closure spanners (TC-spanners).} As implied by \Cref{lem:redTCtoSh}, the hardness result for this parameter range would imply the same hardness for TC-spanners as well.

\paragraph{Label Cover.} Our hardness result is based on a reduction from the well-known \labcov{} problem  to \os{}. We now briefly recall what \labcov{} is: Given an input instance denoted by $\I = (G,\L, \pi)$, where $G=(A \cup B,E)$ is a bipartite undirected graph, $\cL$ is an alphabet, and $\pi = \{\pi_e \subseteq  \L \times \L\}_{e \in E}$ is a family of relations defined for each edge $e \in E$, we say, for any assignment (labeling)  $\psi: A \cup B \rightarrow \L$, that $e=(u,v)$ is covered if $(\psi(u), \psi(v)) \in \pi_e$.

The objective of this problem is to compute a labeling $\psi$ that covers as many edges as possible. We sometimes allow the labeling function to assign more than one label to vertices, i.e., a multi-labeling $\psi: A \cup B \rightarrow 2^{\L}$ covers edge $(u,v)$ if $\psi(u) \times \psi(v)$ contains at least one pair of labels in $\pi_{(u,v)}$. The cost of such multi-labeling is the total number of labels assigned, that is, $\sum_{u \in A \cup B}|\psi(u)|$. 
Assuming the projection games conjecture (PGC)~\cite{Moshkovitz15}, it is hard to distinguish between the following two cases: In the completeness case, there exists an assignment $\psi$ that covers every edge,  while in the soundness case, we cannot cover more than $|\I|^{-\epsilon}$  fraction of edges even when using a multi-labeling $\psi:A \cup B \rightarrow 2^{\L}$ of cost $|\I|^{\epsilon}(|A|+|B|)$.

\paragraph{Initial attempt (\cref{subsec:warmup}).} We start from a simple (base) reduction that takes a label cover instance $\I$ and a target diameter $\rho$  as input and produces a directed graph $H_{\I, \rho}$, or $H$ is short,  with $|V(H)| = \poly(\I)$, such that the following happens (see~\cref{fig:minrepgraph}):
\begin{enumerate}
    \item[-] \textbf{Bijection:} There is a special set of pairs of vertices $C_H \subseteq V(H) \times V(H)$, denoted as the \textit{canonical set}, such that any subset of $C_H$ corresponds to a multi-labeling in $\I$ and vice versa.

    \item[-] \textbf{Completeness:} A multi-labeling in $\I$ would precisely correspond to a shortcut solution that chooses only the pairs in $C_H$; moreover, if the multi-labeling covers all edges, then the corresponding shortcut will reduce the diameter of $H$ to be at most $d=\rho+1$\yonggang{Is it a good idea to define $d=\rho+1$? We also use $d_H$ to refer to the density later.}\yonggang{It would be good to just say $\rho$, but in the construction of $H_{I,\rho}$, the truth is $d=\rho+1$. Now I realize that I should really define $H_{I,\rho}=H_{I,\rho-1}$}. This means any optimal perfectly covering labeling $\psi$ for $\I$ corresponds to feasible shortcuts $C_{\psi} \subseteq C_H$ of exactly the same cost. 

    \item[-] \textbf{Soundness:} Conversely, we prove that any shortcut solution $C'$ that reduces the diameter of $H$ to $d$ corresponds to a multi-labeling of roughly the same size (thereby, corresponds to a \textit{different} shortcut within $C_H$ of roughly the same size), while covering at least $|\I|^{-\epsilon}$ fraction of edges in $\I$.
\end{enumerate}

To summarize, this reduction establishes that the optimal shortcut value in $H$ is roughly the same as the optimal cost of multi-labeling in $\I$. In the completeness case, we therefore are guaranteed to have a shortcut set of cost $|A|+|B|$, while in the soundness case, there is no feasible shortcut set of cost less than $|\I|^{\epsilon}(|A|+|B|)$. This would imply a factor of $|\I|^{\epsilon} \approx |V(H)|^{\Omega(\epsilon)}$ hardness.

The construction is similar to the previous work for proving the hardness of approximating directed spanner~\cite{ElkinP07}. 
However, it has two major deficiencies. \begin{enumerate}[(i)]
    \item \textbf{Small shortcut set:} Recall from the statement of \Cref{thm:intro:lowerbound} that we need $s = \omega(|E(H)|)$. This is crucial to transfer our shortcut lower bound to TC-spanner lower bound (see \Cref{lem:redShtoTC}). However, by construction, the canonical set $C_H$, and thereby the shortcut set, is much smaller than $|E(H)|$, i.e., $s \leq |C_H| = o(|E(H)|)$. Phrased somewhat differently, by construction, the density of the canonical set, $d_H  =|C_H|/|E(H)| = o(1)$ is vanishing which we cannot tolerate.
    
    \item \textbf{Small diameter approximation:} The reduction only rules out $(O(1), n^{\epsilon})$-approximation algorithms under PGC (i.e. $\alpha_D$ is very small). That is unavoidable if we use  constructions similar to $H$ (see~\cref{fig:minrepgraph}), because $H$ itself has diameter $O(d)$ without adding any shortcut. To get a hardness result for much larger $\apxD$, we need a construction where the shortcuts can reduce the diameter significantly. 
\end{enumerate} 
Overcoming these two deficiencies is the main technical contribution of this paper.

\paragraph{Boosting canonical set density (\cref{subsec:largebicriteria}).} In order to overcome the first drawback, we want, as a necessary condition, a construction that boosts the density $d_H$ of the canonical solution. One natural idea to this end (which is also implicitly used by~\cite{BhattacharyyaGJRW12}) is to ``compose'' the base construction $H$ with some ``combinatorial object'' $O$ that has desirable properties (for instance, one such property is that the density of canonical pairs inside $O$ should be high, that is, $\omega(1)$.).
In particular, the reduction takes \labcov{} instance $\I$ and outputs $H^* = H \odot O$ (where $\odot$ represents some kind of composition between $H$ and $O$ that would not be made explicit here). 
The object $O$ would be chosen so that the density $d_{H^*} = |C_{H^*}|/|E(H^*)|$ can be made $\omega(1)$. In this way, we have $|C_{H^*}|=\omega(|E(H^*)|)$ which means $s$ can be $\omega(m)$.\footnote{In fact, the value of $s$ can be much less than the size of the canonical set. For convenience of discussion, this overview focuses on boosting the size of the canonical set (which is, strictly speaking, still not sufficient to prove our hardness result.)} %

Let us denote by $\gamma_O = d_{H\odot O}/d_H$ the (effective) \textbf{boosting factor} of the object $O$; this parameter reflects how much we can increase the contribution of the canonical set after composing the base construction with $O$. If we can make $\gamma_O$ to be sufficiently high, this would allow us to obtain the desired result for large shortcut set.  %

In the case of~\cite{BhattacharyyaGJRW12}, the properties of the combinatorial object they need are provided by the butterfly graph (denoted by $O_B$), whose effective boosting factor decreases exponentially in the target diameter $d$, that is $\gamma_{O_B} = O(n^{1/d})$ (or equivalently, the butterflies grow exponentially, i.e., $|O_B|=\gamma_{O_B}^{\Omega(d)}$), and therefore, when $d = \Omega(\log n)$, such a construction faces its theoretical limit (the boosting factor is a constant), obtaining only an NP-hardness result. We circumvent this barrier via two new ideas: 

\begin{itemize}
    
    \item \textbf{Intermediate problem(see~\cref{def:gadget})}: We introduce a ``Steiner'' variant of the \os{} problem as an intermediate problem (that we call the \textit{minimum Steiner shortcut problem} ({\sf MinStShC})). In this variant, we are additionally given, as input, two disjoint sets $L, R \subseteq V(H): L \cap R = \emptyset$, and a set of pairs $P\subseteq L\times R$. Our goal is to distinguish whether (i) adding a size $s$ shortcut reduces the distance between any pairs $(u,v) \in L \times R$ to constant, or (ii) adding a size $sn^{\epsilon}$ shortcut cannot reduce the distance of even $o(1)$ fraction of pairs in $P$ to $d/3$. We shortly explain why we consider this variant.

    \item \textbf{More efficient combinatorial object}: Our goal now turns into canonical boosting for {\sf MinStShC} (instead of the \os{} problem) by composition with a combinatorial object having the desired properties. Our combinatorial object is inspired by the techniques of Hesse~\cite{Hesse03}, Huang and Pettie~\cite{HuangP21}\footnote{We note here that even though \cite{BhattacharyyaGJRW12} realised the upshot of using the gadget from \cite{Hesse03} regarding canonical boosting, they need a gadget with more structure in order for their technique to work, which is the butterfly graph. Because of our use of the intermediate problem, we do not face this bottleneck.}. It ensures weaker combinatorial properties than the butterflies but would still be sufficient for our purpose (crucially because we work with the intermediate problem). We denote our object by $O^*$. The main advantage of $O^*$ is its relatively compact size, i.e., to obtain a boosting factor of $\gamma_{O^*}$, the size of $O^*$ required is only $|O^*| \leq {\rm poly}(\gamma_{O^*})$. %

\end{itemize}

We briefly explain our rationale why we work with the intermediate problem {\sf MinStShC} now. Our use of {\sf MinStShC} allows a \textit{simpler and cleaner canonical boosting process}. For example, \cite{BhattacharyyaGJRW12} needs to modify the \labcov{} instance to be ``noise-resilient''  (and therefore their reduction adds an extra pre-processing step that turns any \labcov{} instance $\I$ into a noise-resilient instance $\I'$ which is suitable for performing canonical boosting). In contrast, our reduction works with  any given \labcov{} instance $\I$ in a blackbox fashion. \yonggang{I dont't understand this paragraph up to here}%
We remark that, unlike \cite{BhattacharyyaGJRW12}, our canonical boosting process would not work with the original \os{} problem, so the use of the intermediate problem is really crucial for us. 

As our ultimate goal is to obtain hardness for \os{}, we somehow need to convert the hardness of {\sf MinStShC} to \os (We explain more on this in the next paragraph.). To accomplish this, we need \textit{the gap version} of {\sf MinStShC} hardness (for technical reasons). This is not a significant concern in the reduction from $\labcov{}$ to {\sf MinStShC}, because the hardness for \labcov{} is itself a gap version.

\paragraph{From {\sf MinStShC} to $(n^{\epsilon}, n^{\epsilon})$ hardness for \os{} (\cref{subsec:bicriterialower}).}
It turns out that our object $O^*$ can be used in two crucial ways: (i) to perform canonical boosting and (ii) to serve as a gadget that turns the hardness of the Steiner variant into the hardness of \os{} in a way that preserves all important parameters (the hardness factor and the relative size of the canonical set) while obtaining the $(n^{\epsilon}, n^{\epsilon})$-bicriteria hardness (boosting the value of $\alpha_D$ from $O(1)$ to $n^{\epsilon}$ as a by-product of this composition). 
In particular, the hard instance $J$ of {\sf MinStShC} can be composed with $O^*$ to obtain the final instance $J \oplus O^*$ for \os{}. 
This  step bears a certain similarity with some known constructions in the literature of hardness of approximation, e.g.,~\cite{guruswami1999near,chuzhoy2009polynomial} for directed disjoint path problems and the use of graph products~\cite{chalermsook2014pre,chalermsook2013graph}. However, the combinatorial objects necessary (for the composition step) are often problem-specific, and therefore these works are all technically very different  from this paper.

To summarize, our reduction takes \labcov{} instance $\I$ and produces a base instance $H$ of the Steiner shortcut problem. After that, we perform the canonical boosting by composing $H$ with $O^*$, obtaining the hard instance $H \odot O^*$ of the {\sf MinStShC} with our desired density parameter. Finally, we perform another composition to turn the instance $H \odot O^*$ into the instance $(H\odot O^*) \otimes O^*$ of \os{} (again $\otimes$ denotes a certain composition between instances that would not be made explicit in this section).   
 See Figure~\ref{fig:reduction} for illustration.

\begin{figure}[H]
    \centering
    \includegraphics[width=0.9\textwidth]{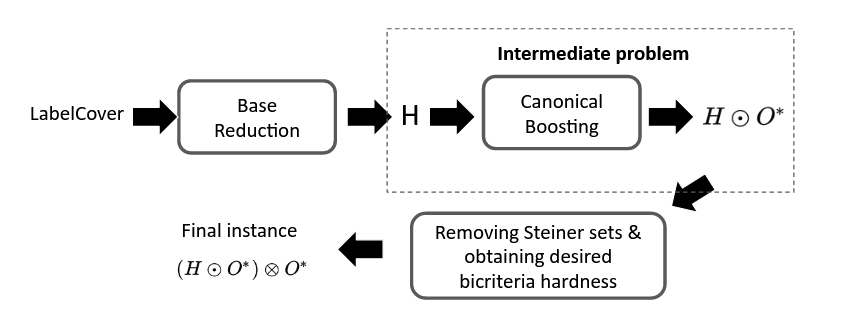}
    \caption{A high-level structure of our reduction. The boosting step  handles the instances of the intermediate problem. }
    \label{fig:reduction}
\end{figure}

%% file: lowerbounds.tex
	\section{Lower Bounds}\label{sec:lowerbounds}

This section contains three parts.

In~\cref{subsec:warmup}, we will prove that there is no polynomial algorithm solving \oss{n^\epsilon}{1}. The technique is similar to the construction used by Elkin and Peleg~\cite{ElkinP07} to prove the hardness of approximating directed spanner. We will also define \ga{}, as discussed in~\cref{sec:overview}, and prove the hardness for \ga{} when $s=o(m)$.

In~\cref{subsec:bicriterialower}, we will show how to prove the lower bound for \oss{n^\epsilon}{n^\epsilon}, using the lower bound for \ga{} while preserving the relative size of $s$ versus $m$. It will only prove that \oss{n^\epsilon}{n^\epsilon} is hard for some input $s=o(m)$ by combining with the result in~\cref{subsec:warmup}, which is still not what we want in~\cref{thm:main}.

In~\cref{subsec:largebicriteria}, we boost $s$ to be $\omega(m)$ in the lower bound proof of \ga{}. Combined with the lemma proved in~\cref{subsec:bicriterialower}, this will show us \oss{n^\epsilon}{n^\epsilon} is hard for some input $s=\Omega(m^{1+\epsilon})$, which is what we want in~\cref{thm:main}.

\subsection{Warm up: Lower Bounds when \texorpdfstring{$\apxD=1$}{apxD1}}\label{subsec:warmup}

This section will use a simple construction to prove the following lemma.

\begin{lemma}\label{lem:onecriteria}
	Assuming PGC (\cref{con:pgc}), there are no polynomial time algorithm solving \oss{n^{\epsilon}}{1} for some small constant $\epsilon$. 
	
\end{lemma}

The construction relies on the following \minre{} graph.

\begin{definition}[\labcov{} graph]\label{def:minrepgraph}
	Given a \labcov{} instance $\I=(A,B,E,\L,(\pi_e)_{e\in E})$ described in~\cref{def:labcov} and a parameter $\rho$, we define the \labcov{} graph $G_{\I,\rho}$ as a directed graph defined as follows (see \cref{fig:minrepgraph}):
	\begin{itemize}
		\item Suppose $A=\{1,...,|A|\},B=\{1,...,|B|\},\L=\{1,2,...,|\L|\}$. In $G_{\I,\rho}$ we have vertices $\{\aa{i}{j},\bb{i}{j}\mid 1\le i\le |A|=|B|, 1\le j\le |\L|\}$ and edges $\{(\aa{i}{j},\bb{i'}{j'})\mid (i,i')\in E, (j,j')\in\pi_{(i,i')}\}$.
		
		\item In addition, $G_{\I,\rho}$ also contains vertices $\{\a{i},\b{i}\mid 1\le i,j\le |A|=|B|\}$, where $\a{i}$ has a directed path with length $\rho$ to each vertex in $\{\aa{i}{j}\mid 1\le j\le |\L|\}$ denoted by $\alpha^{(i)}_j$, the vertices along the path are $\{\aaa{i}{j}{k}\mid 1\le k\le \rho-1\}$; similarly, $\b{i}$ has a directed path with length $\rho$ from each vertex in $\{\bb{i}{j}\mid 1\le j\le |\L|\}$ denoted by $\beta^{(i)}_j$, the vertices along the path are $\{\bbb{i}{j}{k}\mid 1\le k\le \rho-1\}$. 
		
		\item Finally, we add edges $(\aaa{i}{j}{k},\a{i}),(\b{i},\bbb{i}{j}{k}),(\aa{i}{j},\a{i})$, $(\b{i},\bb{i}{j})$ for any possible $i,j,k$.
	\end{itemize}
	
\end{definition}

\begin{figure}[H]
	\centering
	\includegraphics[scale=1]{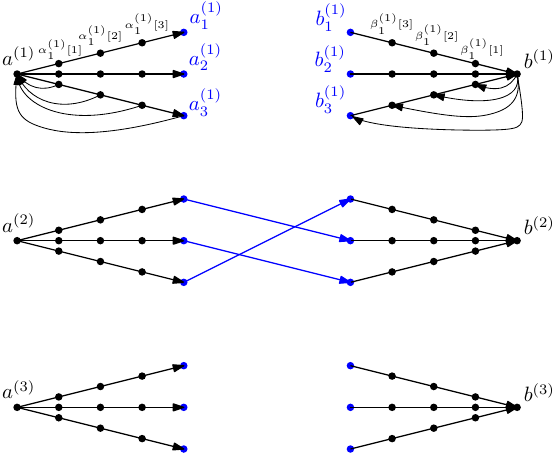}
	\setcaptionwidth{0.95\textwidth}
	\caption{\small In this example, $|A|=|B|=|\mathcal{L}|=\rho-1=3$. Suppose $A=\{1,2,3\},B=\{1,2,3\},\mathcal{L}=\{1,2,3\}$, then in this example we have $E=\{(2,2)\}$ and $\pi_{(2,2)}=\{(1,2),(2,3),(3,1)\}$, which corresponds to the three edges in the middle. }\label{fig:minrepgraph}	
\end{figure}

\begin{proof}[Proof of~\cref{lem:onecriteria}]
	Given a \labcov{} instance $\I=(A,B,E,\L,(\pi_e)_{e\in E})$, we use $\AB$ to denote the size of $A$ (and also $B$). We now describe how to distinguish between the case. Notice that $N$ is the input size of the \labcov{} instance $\I$.
	\begin{itemize}
		
		\item (Completeness:) There is a labeling that covers every edge. 
		
		\item (Soundness:) Any multilabeling of cost at most $N^{\epsl{}}(|A|+|B|)$ covers at most $N^{-\epsl{}}$ fraction of edges.             
	\end{itemize}
	According to~\cref{lem:pgc}, this violates PGC. 
	
	If $|E|\le 2\AB n^{\epsl}$, then it cannot be the Soundness case, so we assume $|E|>2\AB n^{\epsl}$. 
	
	First, we use polynomial time to compute the \labcov{} graph $G_{\I,\dia{}}$ with $n$ vertices, where $\dia{}$ is polynomial in $N$ and can be arbitrarily large. 
	Then we run the \oss{n^\epsilon}{1} algorithm with input $G_{\I,\dia{}}$ and parameters $s=2\AB,d=\dia{}+1$. We will prove the following two claims, which will lead to the solution to the \labcov{} problem.
	\begin{claim}\label{clai:caseI}
		If the \labcov{} instance is in case (completeness), then $G_{\I,\dia{}}$ has a \ssss{2\AB }{\dia{}+1}. Moreover, for any $(i,j)\in E$, $(\a{i},\b{j})$ has distance at most $3$ after adding this \ss{}. 
	\end{claim}
	\begin{proof}
		Suppose the labeling is $\psi$. For any $i\in A$, we include the edge $(\a{i},\aa{i}{\psi(i)})$ in the \ss{}; for any $i\in B$, we include the edge $(\bb{i}{\psi(i)},\b{i})$ in the \ss{}. The \ss{} has size $2\AB $. Since $(i,j)$ is covered by $\psi$, there exists an edge $(\aa{i}{i'},\bb{j}{j'})$ such that $i'=\psi(i),j'=\psi(j)$. In that case, we have the length 3 path $(\a{i},\aa{i}{i'},\bb{j}{j'},\b{j})$ after adding the \ss{}. This proves the second statement of this lemma. Now we verify the distances between reachable pairs are at most $\dia{}+1$ one by one.
		\begin{enumerate}
			\item (start from $\a{i}$) $\a{i}$ can reach any $\aa{i}{j},\aaa{i}{j}{k}$ with distance at most $\dia{}$, $\a{i}$ can reach any $\b{j}$ with $(i,i')\in E$ with distance $3$, which means $\a{i}$ can reach any $\bbb{i'}{j}{k},\bb{i'}{j}$ with distance $4$.
			
			\item (start from $\aaa{i}{j}{k},\aa{i}{j}$) $\aaa{i}{j}{k}$ or $\aa{i}{j}$ has an edge to $\a{i}$, so they can reach any node that $\a{i}$ can reach with distance at most $\dia{}+1$. 
			
			\item (start from $\bbb{i}{j}{k},\bb{i}{j},\b{i}$) These nodes can reach any reachable nodes with distance at most $\dia{}+1$ in $G_{\I,\dia{}}$. 
		\end{enumerate}
		
	\end{proof}
	
	\begin{claim}\label{clai:caseII}
		If the \minre{} instance is in case (soundness), then $G_{\I,\dia{}}$ does not have a \ssss{2\AB\cdot n^{\epsilon}}{\dia{}+1} for sufficiently small constant $\epsilon$. Moreover, by adding any shortcut with size at most $2\AB\cdot n^{\epsilon}$, at most $o(1)$ fraction of pairs in $\{(a^{(i)},b^{(j)})\mid (i,j)\in E\}$ will have distance less than $\dia{}+1$. 
	\end{claim}
	\begin{proof}
		Suppose $G_{\I,\dia{}}$ has a shortcut $E'$ adding which reduces the diameter between $\Omega(1)$ fraction of pairs in $\{(a^{(i)},b^{(j)})\mid (i,j)\in E\}$ to less than $\dia{}+1$. We will try to get a contradiction. We first use $E'$ to build a \mlab{} $\psi$ in the following way: for any $1\le i\le |A|$, we let $\psi(i)$ to be the set of vertices among $\{\aa{i}{j}\mid 1\le j\le \L\}$ that $\a{i}$ has distance at most $\dia{}-1$ to after adding the \ss{}. 
		
		\paragraph{$\psi$ covers more than $N^{-\epsl}$ fraction of edges.}	We first argue that $\psi$ covers $\Omega(1)$ fraction of the edges. Write $A_i=\{\a{i},\aa{i}{j},\aaa{i}{j}{k}\mid 1\le j\le \L,1\le k\le \dia{}-1\}, B_i=\{\b{i},\bb{i}{j},\bbb{i}{j}{k}\mid 1\le j\le \L,1\le k\le \dia{}-1\}$. For an edges $(i,j)\in E$ (recall that $E$ is the edge set in the \labcov{} instance $\I$), we say it is \emph{crossed} if there is an edge $(u,v)\in E'$ with $u\in A_i,v\in B_j$ in $E'$. Remember that we assumed $|E|\ge \Delta\cdot N^{\epsl{}}$, otherwise, the case (soundness) can never happen. Therefore, the number of crossed edges is at most $2\AB\cdot n^{\epsilon}=o(|E|)$ for sufficiently small $\epsilon$. Now We prove that for any non-crossed edge $(i,j)\in E$ such that $(a^{(i)},b^{(j)})$ has distance less than $\dia{}+1$ after adding $E'$, $(i,j)$ is covered by $\psi$. If we can prove this, then at least $\Omega(1)$ fraction of edges in $E$ are covered. To prove this, notice that if $(i,j)$ is not covered, then consider the shortest path $p$ from $\a{i}$ to $\b{j}$ after adding the \ss{}, we write the part where $p$ is inside $A_i$ as $p_A$, and the part where $p$ is inside $B_j$ as $p_B$. if both $p_A,p_B$ have length at most $\dia{}-1$, then $(i,j)$ is covered. Thus, one of $p_A$ or $p_B$ has length at least $\dia{}$. Then we have $|p|=|p_A|+1+|p_B|\ge \dia{}+1+1\ge \dia{}+2$, which is a contradiction. 
		
		\paragraph{$\psi$ has cost at most $N^{\epsl}(|A|+|B|)$.} Then we argue that $\sum_{u\in A\cup B}|\psi(u)|\le |E'|\le 2\AB \cdot n^{\epsilon}$ (which will give us $\sum_{u\in A\cup B}|\psi(u)|\le N^{\epsl{}}(|A|+|B|)$ for sufficiently small $\epsilon$). For a label $j\in\psi(i)$, we have $\distt{(V,E\cup E')}{\a{i},\aa{i}{j}}<\dia{}$. Thus, there exists a path $p=(v_0,...,v_{\ell})$ from $\a{i}$ to $\aa{i}{j}$ with length at most $\dia{}-1$. Let $p'=(\a{i},\aaa{i}{j}{1},\aaa{i}{j}{2},...,\aa{i}{j})$. Let $k$ be the last index such that $v_k\not\in p'$ and $v_{k+1}\in p'$. Since $v_{\ell}$ in $p'$ and the length of $p'$ equals $\dia{}$, such $k$ must exist. Since $(v_k,v_{k+1})\not\in E$, we have $(v_k,v_{k+1})\in E'$. We call this edge $(v_k,v_{k+1})$ a \emph{critical edge} of $(i,j)$. For different $i$ and $j$ with $j\in\psi(i)$, they have different critical edges in $E'$ because the corresponding $v_{k+1}$ is always different. Thus, $\sum_{u\in A\cup B}|\psi(u)|\le |E'|\le 2\AB \cdot n^{\epsilon}$.
		
	\end{proof}
	
	If the \labcov{} instance is in case (completeness), by~\cref{clai:caseI}, $G_{\I,\dia{}}$ admits a \ssss{2\Delta}{\dia{}+1}, which means the output by the \oss{n^{\epsilon}}{1} algorithm will be a \ssss{2\Delta\cdot n^{\epsilon}}{\dia{}+1}; on the other hand, if the \labcov{} instance is in case (soundness), by~\cref{clai:caseII}, the output cannot be a \ssss{2\Delta\cdot n^{\epsilon}}{\dia{}+1} since $G_{\I,\dia{}}$ does not have one. Since checking whether the output is \ssss{2\Delta\cdot n^{\epsilon}}{\dia{}+1} or not is in polynomial time, \labcov{} is solved in polynomial time.
 
\end{proof}

In the following definition, we will abstract all the necessary properties of graph $G_{\I,\rho}$ in a black box that we are going to use to prove the hardness for \oss{n^{\epsilon}}{n^\epsilon}.

\begin{definition}\label{def:gadget}
	For $1>\epsilon,\cs{}>0$, the \gadget{\cs{}}{\epsilon} problem has inputs
	\begin{enumerate}
		\item a directed connected graph $G=(V,E)$ with $m$ edges, $n$ nodes and diameter polynomial in $m$, let the diameter be $d$,
		\item two sets $L,R\subseteq V$ with $L\cap R=\emptyset$, where $|L|$ is polynomial in $m$,
		\item a set of reachable vertex pairs $P\subseteq L\times R$,
		\item a positive integer $s=\Omega(m^{\cs{}})$.
	\end{enumerate}
	The problem asks to distinguish the following two types of graphs.
	\begin{description}
		\item[Type 1.] There exists a \ss{} $E'$ of $G$ with size $s$ such that all reachable pairs $(u,v)\in L\times R$ have distance $O(1)$ after adding $E'$ to $G$.
		\item[Type 2.] By adding any \ss{} with size $s\cdot n^\epsilon$, at most $o(1)$ fraction of the pairs in $P$ have distance at most $d/3$.
	\end{description}
	
\end{definition}

The following lemma shows that~\cref{fig:minrepgraph} is a hard instance for \gadget{\cs{}}{\epsilon}.
\begin{lemma}\label{lem:smallgadget}
	Under PGC (\cref{con:pgc}), there exist constants $\epsilon,\cs{}$ such that \gadget{\cs{}}{\epsilon} cannot be solved in polynomial time.
	
\end{lemma}
\begin{proof}
	Given any \labcov{} instance $\I=(A,B,E,\L,\pi)$ (we write $|A|=\AB$), we construct a \gadget{\cs{}}{\epsilon} instance with the following inputs.
	\begin{enumerate}
		\item The graph is $G_{\I,\dia{}}$ (see~\cref{def:minrepgraph}) for $\dia{}$ arbitrarily large such that $\dia{}$ is polynomial in $N$ (the bit length of instance $\I$). $G_{\I,d}$ has diameter $d=2\dia{}+1$.
		\item $L=\{a_i\mid 1\le i\le \AB\},R=\{b_i\mid 1\le i\le \AB\}$, clearly $|L|=|R|\le m$ and $|L|$ is polynomial in $m$.
		\item $P=\{(a_i,b_j)\mid (i,j)\in E\}$.
		\item $s=|A|+|B|$, it is polynomial in $m$.
	\end{enumerate}
	
	Now we prove that if we have a polynomial time algorithm to distinguish the two types of graphs as described in~\cref{def:gadget}, then we can solve \labcov{} in polynomial time.
	
	If $\I$ is in case (completeness), according to~\cref{clai:caseI}, by adding a \ss{} with size $s=(|A|+|B|)$, all reachable pairs $(a_i,b_j)$ (with $(i,j)\in E$) has distance at most 3. 
 
 If $\I$ is in case (soundness), according to~\cref{clai:caseII}, by adding a \ss{} with size less than $s\cdot n^{\epsilon}$ for sufficiently small $\epsilon$, at most $o(1)$ fraction of pairs in $P$ has distance less than $(2d+1)/3<\dia{}$. 
\end{proof}

\subsection{Lower Bound when \texorpdfstring{$\apxD>1$}{apxDge1} and \texorpdfstring{$s=o(m)$}{0<cs<1}}\label{subsec:bicriterialower}
In this section, we prove the following lemma, which shows how to use the hardness of \gadget{\cs{}}{\epsilon} to get lower bounds for $\apxD>1$.
\begin{lemma}\label{lem:usegadget}
	For any constant $1>\epsilon,\cs{}>0$, if there is no polynomial algorithm solving \gadget{\cs{}}{\epsilon}, then for sufficiently small constant $\gamma$, there is no polynomial algorithm solving \oss{n^{\gamma\epsilon}}{n^{\gamma\epsilon}} even if the input is restricted to $s=\Omega(m^{1+\gamma(\cs{}-1)})$.
	
\end{lemma}
By combining \cref{lem:smallgadget,lem:usegadget}, we can get the following corollary. Notice that this corollary is not the same as~\cref{thm:main}, since it does not restrict $s$ to be $\omega(m)$.

\begin{corollary}\label{cor:smallmain}
	Under \conj{} (\cref{con:pgc}), there is no polynomial algorithm solving \oss{n^{\epsilon}}{n^{\epsilon}} for sufficiently small constant $\epsilon$.
\end{corollary}

The following geometric graph~\cite{HuangP21} is a crucial part of our proof for~\cref{lem:usegadget}. We say a graph $G=(V,E)$ is a $k$-layered directed graph if $V=V_1\uplus V_2\uplus...\uplus V_k$ ($V$ is partitioned into $V_1,...,V_k$), such that for any edge $(u,v)\in E$, there exists $1\le i<k$ and $u\in V_i,v\in V_{i+1}$. $V_i$ is called the $i$-th layer of $G$.

\begin{lemma}[Section 2.2~\cite{HuangP21}]\label{lem:geograph}
	For any $\AB\in\mathbb{N}^+$, for arbitrary constant $0<\lan{}< 1$, we can compute in polynomial time a $\AB^\lan$ layered graph $G=(V,E)$, a set of pairs $P$ and an indexing $\idx{}$ satisfying the following properties.
 
	\begin{enumerate}
            \item Each layer of $G$ contains $\The{\AB^{10}}$ vertices and has a maximum in-degree and out-degree of at most $\AB$. Let the first layer be denoted as $V_1 = \{s_1, s_2, \dots, s_{K_1}\}$ and the last layer as $V_2 = \{t_1, t_2, \dots, t_{K_2}\}$. \label{geoitem1}

		\item $P \subseteq [K_1] \times [K_2]$. For every $i\in[K_1]$, there are $\Omega(\Delta^2)$ pairs $(i,j)\in P$. For every $(i,j)\in P$, there is a unique path from $s_i$ to $t_j$ in $G$, denoted by $p_{i,j}$. $p_{i,j}$ and $p_{i',j'}$ shares at most one edge for different pairs $(i,j)\not=(i',j')$. \label{geoitem3}
		\item The function $\idx{}$ assigns a value from $[\Delta]$ to each $(v, e)$ pair, where $e$ is either an in-edge or an out-edge of vertex $v$. For a given vertex $v$, $\idx{}$ assigns distinct values to its different in-edges, and similarly, $\idx{}$ assigns distinct values to its different out-edges.  For any $(x,y)\in[\AB]\times[\AB]$, there exists $\Omega(\Delta^{10})$ paths $p_{i,j}$ (where $(i,j)\in P$) such that for any edge $(u,v)$ on $p_{i,j}$,\label{geoitem4}
                \begin{itemize}
                    \item if $v$ is in the even layer, then $\Idx{v,(u,v)}=x$, 
                    \item if $u$ is in the even layer, then $\Idx{u,(u,v)}=y$.
                \end{itemize}
	\end{enumerate}
 
\end{lemma}
\begin{proof}[Proof of~\cref{lem:geograph}]
	The general idea is to use the graph described in Section 2.2~\cite{HuangP21} with max degree $\Delta$, and cut the first $\Delta^{\lan{}}$ layers to get our desired graph. Readers are recommended to read the construction in Section 2.2~\cite{HuangP21}, and the construction for~\cref{lem:geograph} is followed straightforwardly. 
	
    To be more specific, we provide a brief guide on how to modify Section 2.2 of \cite{HuangP21} to obtain the desired lemma. We adopt the notations introduced in Section 2.2 of \cite{HuangP21}; therefore, readers are encouraged to refer to that section for a more detailed exposition. Moreover, the properties stated in \Cref{lem:geograph} are standard, so once Section 2.2 of \cite{HuangP21} has been reviewed, the validity of the lemma should be evident.
    
    Let the graph $G_1\otimes G_2$ described in Section 2.2~\cite{HuangP21} with parameters $d_1=d_2=2,r_1=\Delta^{3/2}$ be $G_{\Delta}$. $G_{\Delta}$ has $\Omega(\Delta)$ layers, where each layer has $\The{\Delta^{10}}$ vertices. The max in/out-degree is $\Delta$. We use the subgraph of $G_{\Delta}$ induced by the first $\Delta^{\lan{}}$ layers
	to get the desired graph described in~\cref{lem:geograph}, denoted by $G'_{\Delta}$. %
	We construct $P$ in the following way. According to Section 2.2~\cite{HuangP21}, there exist $\Omega(\Delta^{12})$ \emph{Critical Pairs} in $G_{\Delta}$ between the first and last layer nodes where each of them has a unique path connecting them. We cut each unique path at the first $\Delta^{\lan{}}$ layers, resulting in $\Omega(\Delta^{12})$ pairs between the first and last layer in $G'_{\Delta}$. %
	Now we verify the properties of $P$ one by one.
	\begin{enumerate}
		\item In Section 2.2~\cite{HuangP21}, each critical pairs is $((a,b,0),(a+Dv,b+Dw,2D))$ for some $v,w\in\mathcal{V}_2(r_1)$. There are $\Delta$ possible choices for $v$ or $w$. Thus, for every $a$, there are $\Omega(\Delta^2)$ choices of $b$ such that $(a,b)\in P$.
		\item If the path between $(a,b)\in P$ is not unique, then the path between two critical pair in $G_{\Delta}$ is also not unique.
		\item Since every two paths between two critical pairs in $G_\Delta$ share at most one edge according to Lemma 2.6~\cite{HuangP21}, our path segments in the first $\Delta^{\lan{}}$ layers also share at most one edge.
	\end{enumerate}

        Now we construct $\idx{}$. We fix an order for the set $\mathcal{V}_2(r_1)=\{w_1,w_2,...,w_\Delta\}$. Each node in the even layer $(a,b,i)$ has out-edges $\{(a+w_j,b,i+1)\left((a,b,i),(a+w_j,b,i+1)\right)\mid j\in[\Delta]\}$. $\idx{}$ assign index $j$ to the edge $\left((a,b,i),(a+w_j,b,i+1)\right)$. $(a,b,i)$ has in-edges $\{\left((a,b-w_j,i-1),(a,b,i)\right)\mid j\in[\Delta],b-w_j\in B_2(R_1+\left\lceil (i-1)/2\right\rceil r_1)\}$. $\idx{}$ assign $j$ to the edge $\left((a,b-w_j,i-1),(a,b,i)\right)$. Now for any $(x,y)\in[\Delta]\times[\Delta]$, we consider all the unique paths from $(a,b,0)$ to $(a+\Delta^{\lan{}}x,b+\Delta^{\lan{}}y,\Delta^{\lan{}}-1)$. All edges in this path are either $\left((a+ix,b+iy,2i),(a+(i+1)x,b+iy,2i+1)\right)$, which is the in-edge of $(a+(i+1)x,b+iy,2i+1)$ assigned as $x$, or $((a+(i+1)x,b+iy,2i+1)$,$(a+(i+1)x,b+(i+1)y,2i+2))$, which is the out-edge of $(a+(i+1)x,b+iy,2i+1)$ assigned as $y$.\footnote{To avoid confusion, notice that in Section 2.2~\cite{HuangP21}, the node $(x,y,i)$ is in the $(i+1)$-th layer, because $i$ is indexed from $0$. }
\end{proof}
Now we are ready to prove the main lemma in this section.
\begin{proof}[Proof of~\cref{lem:usegadget}]
	Suppose $\mathcal{A}$ is the polynomial time algorithm for \oss{n^{\gamma\epsilon}}{n^{\gamma\epsilon}} where input must satisfy $s=\Omega(m^{1+\gamma(\cs{}-1)})$. Given a \gadget{\cs{}}{\epsilon} instance $G_{inr},(L,R),P',\bufs$ (see~\cref{def:gadget}, we use $G_{inr},P',s'$ instead of $P,s$ to avoid conflicting of notations), we will show how to use $\mathcal{A}$ as an oracle to solve it in polynomial time.
	
	\paragraph{Definition of $G$.} Let $\Delta=|L|=|R|$, let $M$ denote the number of edges in $G_{inr}$ and let the diameter of $G_{inr}$ be $\rho{}$. We first construct a graph $G$ in the following way. Let the graph, pairs, and indexing described in~\cref{lem:geograph} with parameter $\AB$ and sufficiently small constant $\lan$ be $G_{geo},P,I$. Without loss of generality, we assume $\AB^\lan$ is odd. 
	$G$ contains the following parts. See~\cref{fig:origin} as an example.
	
	\begin{figure}
		\centering
		\includegraphics[scale=0.8]{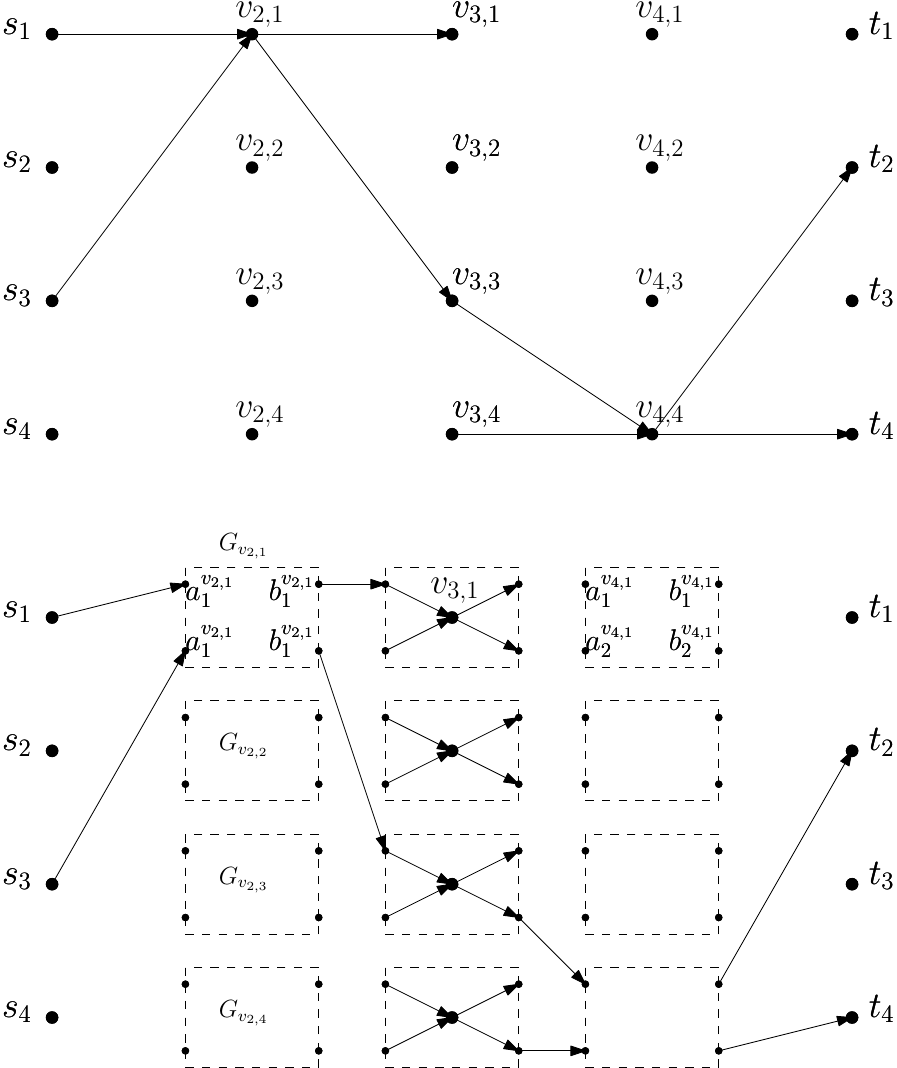}
		\setcaptionwidth{0.95\textwidth}
		\caption{Suppose the graph above is the graph $G_{geo}$ with $\AB=2$ (notice that for ease of explanation, the graph does not satisfy properties described in~\cref{lem:geograph}). The graph below shows how we substitute each node $v$ by $G_v$. If $v$ is in the even layer, then $G_v$ is a copy of $G_{inr}$; otherwise, if $v$ is not in the first or last layer, $G_v$ is a star graph. }\label{fig:origin}	
	\end{figure}
	
	\begin{enumerate}
		\item (Substitute nodes in even layers by copies of $G_{geo}$) for each vertex $v$ in the even layers of $G_{geo}$, create a copy of graph $G_{inr}$, denoted as $G_v$ as part of $G$. Suppose $L=\{a_1,...,a_{\Delta}\},R=\{b_1,...,b_\Delta\}$ (remember that $L,R$ are inputs to \gadget{\cs{}}{\epsilon}). $a_i$ and $b_i$ are denoted by $a^v_i,b^v_i$ in $G_v$. 
		\item (Substitute nodes in odd layers by star graphs) for each vertex $v$ in the odd layers except the first and last layer of $G_{geo}$, create vertices $s_1,...,s_{\Delta},t_1,...,t_\Delta,v$, and create edges $(s_1,v),,...,(s_\Delta,v),(v,t_1),...,(v,t_\Delta)$.
		\item (Keep first and last layer) remember that $I$ is the input to \gadget{\cs{}}{\epsilon}. $G$ includes all nodes in the first and last layer of $G_{geo}$, denoted by $s_1,s_2,...,s_{K_1}$ and $t_1,t_2,...,t_{K_2}$.
		\item (Edges connecting different parts) for each edge $e=(u,v)$ in $G_{geo}$, create an edge $(b^u_{\Idx{u,e}},a^v_{\Idx{v,e}})$ in $G$. If $u$ is in the first layer or $v$ is in the last layer, just create edge $(u,a^v_{\Idx{v,e}})$ or $(b^u_{\Idx{u,e}},v)$ instead.
	\end{enumerate}
	
	Remember that $G_{geo}$ has $\Delta^{\lan{}}$ layers, each with $\Delta^{10}$ vertices, and the max degree of $G_{geo}$ is $\Delta$. The number of edges $m$ in $G$ can be calculated as
	\[m=\The{\Delta^{10+\lan{}}\cdot M}+\The{\Delta^{10+\lan{}}\cdot \Delta}+O(\Delta^{10+\lan{}}\cdot \Delta)=\The{\Delta^{10+\lan{}}\cdot(\Delta+M)}\]
	
	\paragraph{Use $\mathcal{A}$ to solve \gadget{\cs{}}{\epsilon}.} After constructing $G$, we apply $\mathcal{A}$ on $G$ with parameter $s=\AB^{10+\lan}\cdot \bufs, d=C\rho$ for sufficiently large constant $C$. We first argue that $s=\Omega(m^{1+\gamma(\cs{}-1)})$ for sufficiently small $\gamma$ so this is a valid input. Remember that $\bufs$ is one of the inputs to \gadget{\cs{}}{\epsilon} such that $s'=\Omega(M^{\cs{}})$. Also remember in~\cref{def:gadget}, we require $|L|=|R|=\Delta\le M$ to be polynomial in $M$. Denote $M=\Delta^{c_{\Delta}}$. Then we have $m=\The{\Delta^{10+\lan{}+c_\Delta}}$ and $s=\Delta^{10+\lan{}+c_{\Delta}\cs{}}$. Now we have $s/m=\The{\Delta^{c_\Delta(\cs{}-1)}}$, where $\Delta^{c_\Delta}=m^{\gamma}$ for some constant $\gamma$. Thus, we have $s=\Omega(m^{1+\gamma(\cs{}-1)})$.
	
Remember that $\rho$ is the diameter of $G_{inr}$. According to~\cref{def:gadget}, $\rho{}$ is polynomial in $M$. Let $\rho{}= M^{c_D}$. %
Remember our goal is to distinguish the following two types of \gadget{\cs{}}{\epsilon} instances.

\begin{description}
	\item[Type 1.] There exists a \ss{} $E'$ of $G_{inr}$ with size $O(s')$ such that all reachable pairs $(u,v)\in L\times R$ have distance $O(1)$ after adding $E'$ to $G_{inr}$.
	\item[Type 2.] By adding any \ss{} with size $\bufs\cdot M^{\epsilon}$, at most $o(1)$ fraction of the pairs in $P'$ have distance at most $\rho{}/3$.
\end{description}

We prove the following two lemmas to show the output of $\mathcal{A}$ suffices to distinguish whether the \gadget{\cs{}}{\epsilon} instance is in type 1 or type 2. %

\begin{lemma}\label{lem:type1}
	If the \gadget{\cs{}}{\epsilon} instance is type 1, then $G$ has a \ssss{s=\AB^{10+\lan}\cdot \bufs}{O(\rho{})}.
\end{lemma}
\begin{proof}
	For each copy of $G_{inr}$ (denoted as $G_v$) in $G$, we add $\bufs$ edges to the \ss{} to make all reachable pairs $(a^v_{i},b^v_j)$ have distance $O(1)$. The \ss{} has size at most $\AB^{10+\lan}\cdot \bufs$. Consider a reachable $(u,v)$ in $G$, we first argue that they have distance $O(\rho{}+\AB^\lan)$ after adding the \ss{}. To see this, notice that a path connecting $u,v$ can only be in the form of $(u,p_1,p_2,...,p_{\ell},v)$ where $p_i$ is a path inside $G_v$ for some $v$, and $\ell\le\AB^{\lan}$. %
	Each $p_i$ with $1<i<\ell$ will be in the form $(a^v_i,...,b^v_j)$ for some $v,i,j$. Notice that $(a^v_i,b^v_j)$ has distance $O(1)$ after adding the shortcut if $v$ is in the even layer (if $v$ is in the odd layer then there is already an existing length 2 paths connecting $(a^v_i,b^v_j)$). Therefore, $u$ has distance $O(\rho{}+\AB^\lan)$ to $v$. Remember that $\rho{}=M^{c_D}=\Delta^{c_\Delta c_D}$. By setting $\lan{}<c_\Delta c_D$, the distance is at most $O(\rho{})$.
\end{proof}

\begin{lemma}\label{lem:type2}
	If the \gadget{\cs{}}{\epsilon} instance is type 2, then $G$ has no \ssss{n^{\gamma\epsilon}s}{n^{\gamma\epsilon}\rho{}}.
\end{lemma}
\begin{proof}
	
	Before we prove our lemma, we need to do some preparations. %
	Remember that $P'$ is the input to \gadget{\cs{}}{\epsilon} such that for any $(x,y)\in P'\subseteq[\AB]\times[\AB]$, $a^v_x$ can reach $b^v_y$ for any $v$. Also recall that according to~\cref{lem:geograph} item~\ref{geoitem4}, for every $(x,y)\in P'\subseteq [\AB]\times[\AB]$, in $G_{geo}$ there exists $\Omega(\Delta^{10})$ paths $p_{i,j}$ (which is the unique path connecting $s_i,t_j$) that all edges $(u,v)$ in this path is the in-edge indexed by $I$ as $x$ of $v$ if $v$ is in the even layer, or the out-edge indexed by $I$ as $y$ of $u$ if $u$ is in the even layer. That means in $G$, there is a path from $s_i$ to $t_j$ in the form of $p'_{i,j}=(s_i,a^{v_{2,q_2}}_{x},...,b^{v_{2,q_2}}_{y}, a^{v_{3,q_3}}_x,...,b^{v_{3,q_3}}_y,...,t_j)$. We say $p'_{i,j}$ covers vertices $v_{2,q_2},v_{3,q_3},...$ at $(x,y)$. Moreover, for $(i',j')\not=(i,j)$, we know from~\cref{lem:geograph} item~\ref{geoitem3} that $p_{i,j},p_{i',j'}$ shares at most one edge. Thus, $p'_{i,j}$ and $p'_{i',j'}$ cannot cover the same vertex $v$ at the same $(x,y)$, otherwise they share two edges: the in-edge indexed as $x$ of $v$ and the out-edge indexed as $y$ of $v$. In summary, we have $\Omega(\Delta^{10}|P'|)$ pairs $(i,j)$, which we call critical pairs, satisfying the following properties.
	\begin{enumerate}
		\item $s_i$ can reach $t_i$ in $G$, where all paths from $s_i$ to $t_i$ must be in the form $(s_i,a^{v_{2,q_2}}_{x},...,b^{v_{2,q_2}}_{y}, a^{v_{3,q_3}}_x,$ $...,b^{v_{3,q_3}}_y,...,t_j)$. 
		\item Any path from $s_i$ to $t_j$ cover one vertex $v$ in each layer at some pair $(x,y)\subseteq P'$. For different critical pairs $(i,j)\not=(i',j')$, they will not cover the same vertex $v$ at the same pair $(x,y)\subseteq P'$.
	\end{enumerate}

	Now we are ready to prove \cref{lem:type2}. We will prove that $G$ does not have a \ssss{\AB^{10+\lan}\cdot \bufs\cdot M^{\epsilon/2}}{\rho{}\AB^{\lan}/9}. We first show why this implies~\cref{lem:type2}. Remember that $s=\AB^{10+\lan}\cdot \bufs$, thus, $\AB^{10+\lan}\cdot \bufs\cdot M^{\epsilon/2}>m^{\gamma\epsilon}s)$ for sufficiently small constant $\gamma$. Also remember that $m$ is polynomial in $\Delta$, thus, $\rho{}\Delta^{\lan{}}/9>m^{\gamma\epsilon}\rho{}$ for sufficiently small constant $\gamma$. 
	
	Suppose to the contrary, $G$ has a \ssss{\AB^{10+\lan}\cdot \bufs\cdot M^{\epsilon/2}}{\rho{}\AB^{\lan}/9}, we will make a contradiction. We first claim that $G$ also has a \ssss{\AB^{10+\lan}\cdot \bufs\cdot M^{\epsilon/2}\cdot \AB^{\lan}}{\rho{}\AB^{\lan}/8}, denoted by $E'$, such that both end points for every edge in $E'$ is in the same layer. We can achieve this by taking any edge $(u,v)$ in the original \ss{} that crosses different layers, finding the path from $u$ to $v$ denoted by $p$, and cut $p$ into at most $\AB^\lan$ segments $p_1,p_2,...,p_\ell$ where each segment is in the same layer, and replace $(u,v)\in E'$ by $\Delta^{\lan}$ new edges each connecting two end points of $p_i$ from $i=1$ to $i=\ell$. In this way, the \ss{} increase by a factor of at most $\AB^\lan$. The distance between two vertices increases by at most an additive factor of $O(\AB^\lan)$. Another property of $E'$ is that every edge in $E'$ has both endpoints in $G_v$ for some $v$. That is because different $G_v$ where $v$ in the same layer are not reachable from each other. 
	
	Now with $E'$ added to $G$, we denote the new graph by $G'$. We define an indicator variable $I_{v,i,j}$ to be $1$ if $\distt{G'}{a^v_i,b^v_j}\le \rho{}/3$. We know from the property of type 2 that for any $v$ and arbitrary small constant $\eta$, $\sum_{i,j\in P'}I_{v,i,j}>\eta|P'|$ implies $E'$ has at least $\bufs\cdot M^{\epsilon}$ edges in $G_v$. Since $|E'|\le \AB^{10+\lan}\cdot \bufs\cdot M^{\epsilon/2}\cdot \AB^{\lan}$, by setting $\lan<1/2$, at most $o(\AB^{10+\lan})$ vertices $v$ has the property that $\sum_{i,j\in P'}I_{v,i,j}>\eta|P'|$, which means $\sum_{v,(i,j)\in P'}I_{v,i,j}\le 2\eta\AB^{10+\lan}|P'|$. However, for each critical pair $(i,j)$, suppose $T_{i,j}=\{(v,x,y)\mid (i,j)\text{ covers }v\text{ at }(x,y)\}$, then we have $\sum_{(v,x,y)\in T_{i,j}}I_{v,x,y}\ge (1/2)\AB^\lan$, otherwise the distance $\distt{G'}{s_i,t_j}$ would be at least $(\AB^\lan)/2\cdot \rho{}/3>\rho{}\AB^\lan/6$. We also know that $T_{i,j}$ is disjoint for different $i,j$. Thus, we get $\sum_{v,i,j}I_{v,i,j}=\Omega(\Delta^{10}|P'|)\cdot \AB^\lan/2=\Omega(|P'|\AB^{10+\lan})$, contradicts the fact that $\sum_{v,i,j}I_{v,i,j}\le 2\eta\AB^{10+\lan}|P'|$ for sufficiently small constant $\eta$. 
	\yonggang{I think this proof is very hard to parse and I should try to refine it.}
	
\end{proof}

If the \gadget{\cs{}}{\epsilon} instance is type 1, then according to~\cref{lem:type1}, $G$ admits a \ssss{s=\AB^{10+\lan}\cdot \bufs}{O(\rho{})}, so $\mathcal{A}$ will output a \ssss{n^{\gamma\epsilon}s}{n^{\gamma\epsilon}\rho{}}; on the other hand, if the \gadget{\cs{}}{\epsilon} instance is type 2, then according to~\cref{lem:type1}, $G$ cannot output a \ssss{n^{\gamma\epsilon}s}{n^{\gamma\epsilon}\rho{}}. Since checking whether an edge set is a \ssss{n^{\gamma\epsilon}s}{n^{\gamma\epsilon}\rho{}} or not is in polynomial time, the output of $\mathcal{A}$ distinguish type 1 from type 2.

\end{proof}

\subsection{Lower Bound when \texorpdfstring{$\apxD>1$}{paxDge1} and \texorpdfstring{$s=\omega(m)$}{cs>1}}\label{subsec:largebicriteria}

In this section, we prove the following lemma.

\begin{lemma}\label{lem:larges}
	Assuming \conj{} (\cref{con:pgc}), there exist two constants $0<\delta,\epsilon<1$ such that \gadget{1+\delta}{\epsilon} cannot be solved in polynomial time.
\end{lemma}
By combining~\cref{lem:larges,lem:usegadget}, we can get the proof of~\cref{thm:main}.
\begin{proof}[Proof of~\cref{thm:main}]
	According to~\cref{lem:larges}, under \conj{}, there exists two constants $0<\delta,\epsilon<1$ such that \gadget{1+\delta}{\epsilon} cannot be solved in polynomial time. According to~\cref{lem:usegadget}, there are no polynomial algorithm solving \oss{m^{\gamma\epsilon}}{m^{\gamma\epsilon}} when input $s=\Omega(m^{1+\gamma(1+\delta-1)})=\Omega(m^{1+\gamma\delta})$ for some constant $\gamma$. 
\end{proof}

\begin{proof}[Proof of~\cref{lem:larges}]
	Given a \labcov{} instance $\I=(A,B,E,\L,(\pi_e)_{e\in E})$ described in~\cref{def:labcov}, \conj{} implies we cannot distinguish the following two cases in polynomial time for some small constant $\epsl{}$ according to~\cref{lem:pgc}. Remember that $N$ is the input bit length of $\I$.
	
	\begin{itemize}
		\item (Completeness:) There is a labeling that covers every edge. 
		
		\item (Soundness:) Any multilabeling of cost at most $N^{\epsl{}}(|A|+|B|)$ covers at most $N^{-\epsl{}}$ fraction of edges.      	
	\end{itemize}

	Assuming there is a polynomial time algorithm $\mathcal{A}$ solving \gadget{1+\delta}{\epsilon} for sufficiently small constant $\epsilon$, we will show how to distinguish the above two cases in polynomial time, which will lead to a contradiction.

	\paragraph{Definition of $G$.} First, we construct a graph $G$ (see~\cref{fig:large}) according to the instance $\I$. Let $\Delta=N^{c_{\Delta}}$ for a sufficiently large constant $c_{\Delta}$. Denote the graph described by~\cref{lem:geograph} with parameter $\Delta$ as $G_{geo}$. Remember that $G_{geo}$ has $\Delta^{\lan{}}$ layers, each with size $\Theta(\Delta^{10})$, where the first layer is $S=\{s_1,s_2,...,s_{K_1}\}$, the last layer is $T=\{t_1,t_2,...,t_{K_2}\}$.
	\begin{enumerate}
		\item $G$ has two sides: $A$ side and $B$ side. Each side contains $K_1$ \emph{batches}, each batch is composed by $|A|=|B|$ \emph{fans}. Let the reversed graph of $G_{geo}$ (where we reverse all the edge directions in this graph) as $G^R_{geo}$. On the $A$ side, each fan is composed of $\L$ copies of graph $G^R_{geo}$; on the $B$ side, each is composed of $\L$ copies of a graph $G_{geo}$. All $G^R_{geo}$ or $G_{geo}$ in the same fan share the same $T$ vertex set. On the $A/B$ side, the $k$-th copied graph in the $i$-th batch, $j$-th fan is denoted by $G^{A/B}_{i,j,k}$, where the node $s_\ell$ in $G_{geo}$ is denoted by $s^{A/B}_{(i,j),\ell}$ for $1\le \ell\le K_1$ (recall that for different $k$, they share the same $s_\ell$), and the node $t_\ell$ in $G_{geo}$ is denoted by $t^{A/B}_{(i,j,k),\ell}$ for $1\le \ell\le K_2$. 
		\item $G$ has edges defined by the \labcov{} instance $\I$ from $A$ side to $B$ side described as follows. Remember that $\I=(A,B,E,\L,(\pi_e)_{e\in E})$ where $A=\{1,2,...,|A|\},B=\{1,2,...,|B|\},\L=\{1,2,...,|\L|\}$. For every $(j,j')\in E,(k,k')\in\pi_{(j,j')}$ and $i,\ell\in[K_1]$, there is an edge from $s^A_{(i,j,k),\ell}$ to $s^B_{(\ell,j',k'),i}$. 
		An intuition of why the edges are defined in this way is that, now if a path is from the $A$ side of the $i$-th batch to $B$ side of the $\ell$-th batch, the path must go through $s^A_{(i,j,k),\ell}$ to $s^B_{(\ell,j',k'),i}$ for some $(j,j')\in E,(k,k')\in\pi_{(j,j')}$.
		\item For every $\ell\in[K_2]$, create a node $p^A_\ell$ and edges $\{(t^A_{(i,j),\ell},p^A_\ell)\mid 1\le i\le K_1,1\le j\le |A|)\}$. For every $i\in[K_1]$, create a node $q^B_i$ and edges $\{(q^B_i,t^B_{(i,j),\ell})\mid 1\le j\le |A|,1\le \ell\le K_2\}$. Remember that in~\cref{lem:geograph}, $P$ is a set of pairs between first and last layer indexes of nodes in $G_{geo}$. For every $\ell\in[K_2],i\in[K_1]$ where $(i,\ell)\not\in P$, we create an edge $(p^A_{\ell},q^B_{i})$. 
		
		We also create a mirror of the above nodes and edges (\cref{fig:large} does not draw) by reversing $A$ and $B$. i.e., for every $\ell\in[K_2]$, create a node $p^B_\ell$ with edges $\{(p^B_\ell,t^B_{(i,j),\ell})\mid 1\le i\le K_1,1\le j\le |A|)\}$. For every $i\in[K_1]$, create a node $q^A_i$ with edges $\{(t^A_{(i,j),\ell},q^A_i)\mid 1\le j\le |A|,1\le \ell\le K_2\}$. For every $\ell\in[K_2],i\in[K_1]$ where $(i,\ell)\not\in P$, we create an edge $(q^A_{i},p^B_{\ell})$. 
	\end{enumerate} 
	
	Remember that the number of edges in $G_{geo}$ is $M=O(\Delta^{11+\lan{}})$. Let $n$ be the number of nodes in $G$. Let $m$ be the number of edges in $G$. $m$ can be calculated as
	\[m={\color{gray}M\cdot K_1|A||\L|}+{\color{blue}O(|\L|^2|A|^2)\cdot K_1^2}+{\color{orange}4K_1K_2|A|}+{\color{red}O(K_1K_2)}=\Theta(M\cdot K_1|A||\L|)\]
	The third and fourth terms are trivially less than the first term. The second term is less than the first term since we set $\Delta=N^{c_\Delta}$ for sufficiently large constant $c_\Delta$. 
	\begin{figure}
		
		\begin{center}
			
			\begin{tikzpicture}[x={(0.5cm,0.3cm)}, y={(-0.5cm,0.3cm)}, z={(0cm,0.5cm)}]
				
				\newcommand{\con}{3}
				
				\newcommand{\drawpiece}[4]{
					\coordinate (A) at (0 * #1 + #2, -0.3 + #3, 0 + #4);
					\coordinate (B) at (0 * #1 + #2, 2.3 + #3, 0 + #4);
					\coordinate (A1) at (3 * #1 + #2, 2 + #3, 1 + #4);
					\coordinate (B1) at (3 * #1 + #2, 0 + #3, 1 + #4);
					\coordinate (A2) at (3 * #1 + #2, 2 + #3, 0 + #4);
					\coordinate (B2) at (3 * #1 + #2, 0 + #3, 0 + #4);
					\coordinate (A3) at (3 * #1 + #2, 2 + #3, -1 + #4);
					\coordinate (B3) at (3 * #1 + #2, 0 + #3, -1 + #4);
					\draw[fill=gray!30] (A) -- (B) -- (A3) -- (B3) -- cycle;
					\draw[fill=gray!30] (A) -- (B) -- (A2) -- (B2) -- cycle;
					\draw[fill=gray!30] (A) -- (B) -- (A1) -- (B1) -- cycle;
					
				}
				\newcommand{\drawstack}[5]{
					\fill[gray] (1.5 * #1 + #2,1 + #3,-2 + #5  + #4) circle (1pt);
					\fill[gray] (1.5 * #1 + #2,1 + #3,-2.5 + #5  + #4) circle (1pt);
					\fill[gray] (1.5 * #1 + #2,1 + #3,-3 + #5  + #4) circle (1pt);
					\drawpiece{#1}{#2}{#3}{#4 + #5}
					\drawpiece{#1}{#2}{#3}{#4}
				}
				\newcommand{\drawbatch}[6]{
					\fill[gray] (1.5 * #1 + #2,7 + #3,-1.5) circle (1pt);
					\fill[gray] (1.5 * #1 + #2,7.5 + #3,-1.5) circle (1pt);
					\fill[gray] (1.5 * #1 + #2,8 + #3,-1.5) circle (1pt);
					\drawstack{#1}{#2}{#3 + #6}{#4}{#5}
					\drawstack{#1}{#2}{#3}{#4}{#5}
				}
				
				\coordinate (qB1) at (15, 1.5,-1.5);
				
				\draw  (qB1) node[circle, fill, inner sep=1pt, label={[label distance=-0.1cm, color = orange]right:\fontsize{10}{0}\selectfont$q^B_{1}$},  color=orange] {};

				\draw[->, >=stealth, color=orange, line width=0.5pt] (qB1) -- (11,-0.3,0);
				\draw[->, >=stealth, color=orange, line width=0.5pt] (qB1) -- (11,2.3,0);
				
				\draw[->, >=stealth, color=orange, line width=0.5pt] (qB1) -- (11,-0.3,-3);
				\draw[->, >=stealth, color=orange, line width=0.5pt] (qB1) -- (11,2.3,-3);
				\coordinate (qB1) at (15, 1.5,-1.5);
				
				\coordinate (qB2) at (15, 1.5 + 3,-1.5);
				\draw  (qB2) node[circle, fill, inner sep=1pt, label={[label distance=-0.1cm, color = orange]above:\fontsize{10}{0}\selectfont$q^B_{2}$},  color=orange] {};

				\draw[->, >=stealth, color=orange, line width=0.5pt] (qB2) -- (11,-0.3 + 3,0);
				\draw[->, >=stealth, color=orange, line width=0.5pt] (qB2) -- (11,2.3 +  3,0);
				
				\draw[->, >=stealth, color=orange, line width=0.5pt] (qB2) -- (11,-0.3 +  3,-3);
				\draw[->, >=stealth, color=orange, line width=0.5pt] (qB2) -- (11,2.3 +  3,-3);
				
				\drawbatch{-1}{11}{0}{0}{-3}{3}
				
				\draw[blue, dashed, dash pattern=on 3pt off 3pt] (3,0,-6) -- (3,0,1) -- (5 + \con,0,1) -- (5 + \con,0,-6);
				\draw[blue, dashed, dash pattern=on 3pt off 3pt] (3,0.5,-6) -- (3,0.5,1) -- (5 + \con,3,1) -- (5 + \con,3,-6);
				\draw[blue, dashed, dash pattern=on 3pt off 3pt] (3,3,-6) -- (3,3,1) -- (5 + \con,0.5,1) -- (5 + \con,0.5,-6);
				\draw[blue, dashed, dash pattern=on 3pt off 3pt] (3,3.5,-6) -- (3,3.5,1) -- (5 + \con,3.5,1) -- (5 + \con,3.5,-6);
				
				\drawbatch{1}{0}{0}{0}{-3}{3}
				
				\node at (1.5, 1, 0.5) {$G^A_{1,1,1}$};
				
				\node at (1.5, 4, 0.5) {$G^A_{2,1,1}$};
				
				\node at (1.5, 1, -2.5) {$G^A_{1,2,1}$};
				
				\node at (9.5, 1, 0.5) {$G^B_{1,1,1}$};
				
				\draw (3,0,1) node[circle, fill, inner sep=1pt, label={[label distance=-0.1cm, color = blue]right:\fontsize{5}{0}\selectfont$s^A_{(1,1,1),1}$},  color=blue] {};
				
				\draw (3,0.5,1) node[circle, fill, inner sep=1pt, label={[label distance=-0.1cm, color = blue]left:\fontsize{5}{0}\selectfont$s^A_{(1,1,1),2}$},  color=blue] {};
				
				\draw (3,0,0) node[circle, fill, inner sep=1pt, label={[label distance=-0.1cm, color = blue]right:\fontsize{5}{0}\selectfont$s^A_{(1,1,2),1}$},  color=blue] {};
				
				\draw (3,0,-2) node[circle, fill, inner sep=1pt, label={[label distance=-0.1cm, color = blue]right:\fontsize{5}{0}\selectfont$s^A_{(1,2,1),1}$},  color=blue] {};
				
				\draw (3,3,1) node[circle, fill, inner sep=1pt, label={[label distance=-0.1cm, color = blue]left:\fontsize{5}{0}\selectfont$s^A_{(2,1,1),1}$},  color=blue] {};
				
				\draw (0,-0.3,0) node[circle, fill, inner sep=1pt, label={[label distance=-0.1cm, color = blue]above:\fontsize{5}{0}\selectfont$t^A_{(1,1),1}$},  color=blue] {};

				\draw (8,0,1) node[circle, fill, inner sep=1pt, label={[label distance=-0.1cm, color = blue]right:\fontsize{5}{0}\selectfont$s^B_{(1,1,1),1}$},  color=blue] {};
				\draw (11,-0.3,0) node[circle, fill, inner sep=1pt, label={[label distance=-0.1cm, color = blue]right:\fontsize{5}{0}\selectfont$t^B_{(1,1,1),1}$},  color=blue] {};
				
				\coordinate (pA1) at (-4,0,-1.5);
				
				\draw  (pA1) node[circle, fill, inner sep=1pt, label={[label distance=-0.1cm, color = orange]left:\fontsize{10}{0}\selectfont$p^A_{1}$},  color=orange] {};

				\draw[->, >=stealth, color=orange, line width=0.5pt] (0,-0.3,0) -- (pA1);
				\draw[->, >=stealth, color=orange, line width=0.5pt] (0,2.7,0) -- (pA1);
				
				\draw[->, >=stealth, color=orange, line width=0.5pt] (0,-0.3,-3) -- (pA1);
				\draw[->, >=stealth, color=orange, line width=0.5pt] (0,2.7,-3) -- (pA1);

				\coordinate (pA2) at (-4,2,-1.5);
				
				\draw  (pA2) node[circle, fill, inner sep=1pt, label={[label distance=-0.1cm, color = orange]below:\fontsize{10}{0}\selectfont$p^A_{2}$},  color=orange] {};

				\draw[->, >=stealth, color=orange, line width=0.5pt] (0,0.2,0) -- (pA2);
				\draw[->, >=stealth, color=orange, line width=0.5pt] (0,3.2,0) -- (pA2);
				
				\draw[->, >=stealth, color=orange, line width=0.5pt] (0,0.2,-3) -- (pA2);
				\draw[->, >=stealth, color=orange, line width=0.5pt] (0,3.2,-3) -- (pA2);
				\fill[gray] (-4,3,-1.5) circle (1pt);
				\fill[gray] (-4,3.5,-1.5) circle (1pt);
				\fill[gray] (-4,4,-1.5) circle (1pt);

				\coordinate (start) at (-4,0,-1.5);
				\coordinate (end) at (15,4.5,-1.5);
				
				\coordinate (cp1) at (-3,-4,-1.5);
				\coordinate (cp2) at (12,-5,-1.5);
				\coordinate (cp6) at (16,0,-1.5);
				
				\draw[->, color=red, line width=0.5pt] plot[smooth, tension=0.5] coordinates{(start) (cp1) (cp2) (cp6) (end)};

			\end{tikzpicture}
		\end{center}
		\caption{$G^A_{i,j,k}$ is one copy of the graph described in~\cref{lem:geograph} in the $i$-th batch (from right to left), $j$-th fan (from top to bottom) and $k$-th piece. For different $k$ and fixed $i,j$, $G^A_{i,j,k}$ shares the same first layer nodes (which are $s^A_{(i,j),\ell}$), but have different last layer nodes (which are $t^A_{(i,j,k),\ell}$). By changing $A$ to $B$ we get another side of the graph. Each dashed rectangle specifies the graph similar to the middle part in~\cref{fig:minrepgraph} according to the input \labcov{} instance $\I$. Fix $j,k$, for any $i,\ell$, we have $t^A_{(i,j,k),\ell}$ and $t^B_{(\ell,j,k),i}$ in the same dashed rectangle. }\label{fig:large}
	\end{figure}
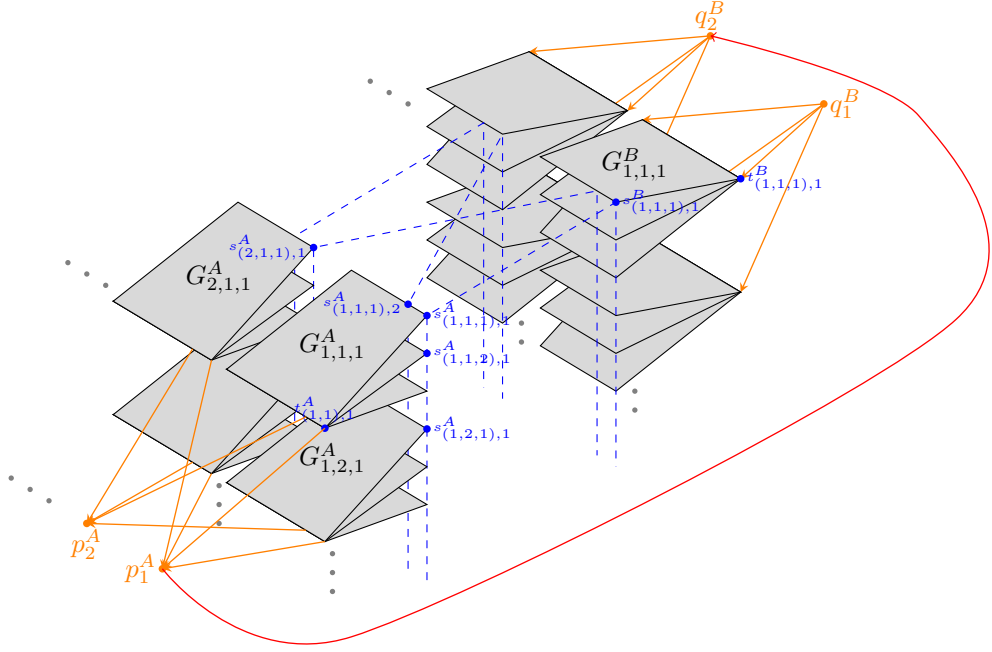

	\paragraph{Solve \labcov{} using \gadget{1+\delta}{\epsilon}} Now we run the \gadget{1+\delta}{\epsilon} algorithm $\mathcal{A}$ with the following inputs. We need to verify that the inputs satisfy the requirements specified by~\cref{def:gadget}.
	\begin{enumerate}
		\item Graph $G$ with $m$ edges. The diameter of $G$ is $d=2\Delta^{\lan{}}$, which is polynomial in $m$.
		\item Two sets $L=\{t^A_{(i,j),k}\mid 1\le i\le K_1,1\le j\le |A|,1\le k\le |\L|\}, R=\{t^B_{(i,j),k}\mid  1\le i\le K_1,1\le j\le |A|,1\le k\le |\L|\}$. The size of each set is $K_1|A||\L|$ which is polynomial in $m$.
		\item A set of reachable vertex pairs $P'\subseteq L\times R$ defined as follows. Recall that in the \labcov{} instance, we have $A=\{1,...,|A|\},B=\{1,...,|B|\}$, and $P$ is the set of pairs defined in~\cref{lem:geograph}. For every $i\in[K_1]$ let $I_{i}$ contain all indexes $x\in[K_2]$ such that $(i,x)\in P$. Let $I_{i}[\ell]$ be the $\ell$-th element in $I_{i}$. 
		Now we define our $P'$.
		\[P'=\left\{\left(t^A_{(i,j),I_{i'}[\ell]},t^B_{(i',j'),I_{i}[\ell]}\right)\mid 1\le i,i'\le K_2,(j,j')\in E,1\le\ell\le\min(|I_{i'}|,|I_{i}|)\right\}\]
		We need to argue that $P'$ only contains reachable pairs. Notice that $t^A_{(i,j),I_{i'}[\ell]}$ can reach $s^A_{(i,j,k),i'}$ for any $k$ since $(i',I_{i}[\ell])\in P'$ (recall that $G^A_{(i,j,k)}$ is the reversed graph $G^R_{geo}$); for the same reason $t^B_{(i',j'),I_{i}[\ell]}$ can be reached from $s^B_{(i',j',k'),i}$ for any $k'$. Now we only need to argue that there exists $k,k'$ such that $t^A_{(i,j,k),i'}$ can reach $s^B_{(i',j',k'),i}$. We can take the $(k,k')\in \pi_{(j,j')}$ where $\pi_{(j,j')}$ must be non-empty since $(j,j')\in E$.
	\end{enumerate}
	Remember that the output of $\mathcal{A}$ will distinguish the following two types of instances. 
	\begin{description}
		\item[Type 1.] There exists a \ss{} $E'$ of $G$ with size $O(m^{1+\delta})$ such that all reachable pairs $(u,v)\in L\times R$ have distance $O(1)$ after adding $E'$ to $G$. %
		\item[Type 2.] By adding any \ss{} with size $O(m^{1+\delta+\epsilon})$, at most $o(1)$ fraction of pairs in $P$ have distance at most $d/3$.  
	\end{description}
	
	Next we show the output can already distinguish the \labcov{} instance (completeness) from (soundness).
	
	\paragraph{(Completeness) implies Type 1.} %
	Suppose the \mlab{} covering all edges in (completeness) is $\psi$. Recall that $P'$ is the pair set described in~\cref{lem:geograph}. We create a \ss{} $E'$ defined as
	\begin{equation*}
		\begin{aligned}
			E' &= \{(t^A_{(i,j),\ell},s^A_{(i,j,k),\ell'}) \mid i\in[K_1],j\in[|A|],k\in\psi(j),(\ell',\ell)\in P'\} \\
			&\cup \{(s^B_{(i,j,k),\ell},t^B_{(i,j),\ell'}) \mid i\in[K_1],j\in[|B|],k\in\psi(j),(\ell,\ell')\in P'\}
		\end{aligned}
	\end{equation*}
	We have $|E'|=O(K_1|A|\Delta^{12})$. Remember that $m=\Theta(M\cdot K_1|A||\L|)$, we have $|E'|=\Theta(m^{1+\delta})$ for some constant $\delta$. Now we prove that all reachable pairs $(t^A_{(i,j),\ell},t^B_{(i',j'),\ell'})$ has distance $O(1)$ after adding $E'$. If $(i',\ell)\not\in P$ or $(i,\ell')\not\in P$, then $t^A_{(i,j),\ell}$ can reach $t^B_{(i',j'),\ell'}$ in $3$ steps. Thus, we only consider the case when $(i',\ell),(i,\ell')\in P$. If $(j,j')\not\in E$, then $(t^A_{(i,j),\ell},t^B_{(i',j'),\ell'})$ is not reachable. Suppose $(j,j')\in E$, let $k\in\psi(j),k'\in\psi(j')$, since $(j,j)$ is covered by $\psi$, there is a blue edge $(s^A_{(i,j,k),i'},s^B_{(i',j',k'),i})$. Besides, we have $(t^A_{(i,j),\ell},s^A_{(i,j,k),i'}),(s^B_{(i',j',k'),i},t^B_{(i',j'),\ell'})\in E'$ due to the fact that $(i',\ell),(i,\ell')\in P$. Thus, the distance between $(t^A_{(i,j),\ell},t^B_{(i',j'),\ell'})$ is $3$. 
	
	\paragraph{(Soundness) implies Type 2.} We will prove it by contradiction. Suppose there exists a \ss{} $E'$ with size $O(m^{1+\delta+\epsilon})=O(K_1|A|\Delta^{12}\cdot n^\epsilon)$, after adding which, not $o(1)$ fraction of pairs in $P$ have distance at most $d/3=(2/3)\Delta^{\lan{}}$. We first turn $E'$ into another \ss{} $E''$ where each shortcut in $E''$ is totally in a copy of graph $G_{geo}$ in the following way. Suppose $(u,v)\in E'$ where $u,v$ are not in the same copy of $G_{geo}$, a path from $u$ to $v$ must cross at most two copies of $G_{geo}$, one on thie $A$ size and another on the $B$ size. We split this path into the two copies, and creat two edges connected two end points of both path
	. $E''$ will have size twice of $E'$, and the distance after adding $E''$ will be at most $(2/3)\Delta^{\lan{}}+2$.
	
	Recall that for every $i\in[K_1]$, we defined $I_{i}$ as all indexes $x\in[K_2]$ such that $(i,x)\in P$. We create \mlab{} $\psi_{i,i',\ell}$ in the following way: for every $j\in A,j'\in B$, we let 
	\[\psi_{i,i',\ell}(j)=\left\{k\in\L\mid\distt{G+E''}{t^A_{(i,j),I_{i'}[\ell]},s^A_{(i,j,k),i'}}<\Delta^{\lan{}}\right\}\]
	\[\psi_{i,i',\ell}(j)=\left\{k\in\L\mid\distt{G+E''}{s^B_{(i',j,k),i},t^B_{(i',j),I_{i}[\ell]}}<\Delta^{\lan{}}\right\}\]
	
	One importance observation is, for every $\psi_{i,i',\ell}$ and $k\in\L$, there is a unique path from $t^A_{(i,j),I_{i'}[\ell]}$ to $t^A_{(i,j,k),i'}$ or from $s^B_{(i',j',k),I_i[\ell]}$ to $t^B_{(i',j'),\ell}$; any two of these unique paths shares at most one edge. The reason is $(i,I_i[\ell])\in P$ (see~\cref{lem:geograph}). As a result, each edge in $E''$ can make at most one pair of vertices $(t^A_{(i,j),I_{i'}[\ell]},s^A_{(i,j,k),i'})$ or $(s^B_{(i',j,k),i},t^B_{(i',j),I_{i}[\ell]})$ to have distance less that the original distance $\Delta^{\lan{}}$. Therefore,
	\[\sum_{i,i'\in[K_1],\ell\le \min(|I_{i'}|,|I_{i}|)}\left(|\psi_{i,i',\ell}(j)|+|\psi_{i,i',\ell}(j)|\right)\le |E''|=O(K_1|A|\Delta^{12}\cdot n^\epsilon)\]
	
	According to~\cref{lem:geograph}, we know $\min(|I_{i'}|,|I_{i}|)=\Omega(\Delta^2)$, which means the number of \mlab{} $\psi_{i,i',\ell}$ is $T=\Theta(K_1^2\Delta^2)=\Theta(K_1\Delta^{12})$. Thus, only $o(K_1\Delta^{12})$ of them will have size at least $N^{\epsl{}}(|A|+|B|)$ for sufficiently small constant $\epsilon$, where $N$ is the input size of $\I$ which is polynomial in $n$. Let $C(\psi_{i,i',\ell})$ be the number of edges covered by $\psi_{i,i',\ell}$ for instance $\I$. We have 
	\[\sum_{i,i',\ell}C(\psi_{i,i',\ell})\le o(K_1\Delta^{12})\cdot |E|+T\cdot o(|E|)\le o(T|E|)\]
	
	One can see that if $\psi_{i,i',\ell}$ does not cover an edge $(j,j')$, then the distance from $t^A_{(i,j),I_{i'}[\ell]}$ to $t^B_{(i',j),I_{i}[\ell]}$ is at least $\Delta^{\lan{}}>(2/3)n^{c_D}$, where $(t^A_{(i,j),I_{i'}[\ell]},t^B_{(i',j),I_{i}[\ell]})\in P'$. That is because the only path from $t^A_{(i,j),I_{i'}[\ell]}$ to $t^B_{(i',j),I_{i}[\ell]}$ must be the concatenation of paths between $t^A_{(i,j),I_{i'}[\ell]},s^A_{(i,j,k),i'}$ and between $s^B_{(i',j,k'),i},t^B_{(i',j),I_{i}[\ell]}$ for some $k,k'\in\L$, where at least one of them has length $\Delta^{\lan{}}$. Therefore, we have at least $(1-o(1))T|E|=(1-o(1))|P'|$ pairs in $P'$ that have distance at least $\Delta^{\lan{}}$, which is a contradiction.%
	
\yonggang{the notations in this proof are disasters, I need to refine it}
\end{proof}

%% file: upper.tex
\section{Upper Bounds}\label{sec:upperbound}
In this section, we will prove \Cref{thm:upperbound}. We restate it here.
\UpperBoundThm*

We use the algorithm idea of a previous spanner approximation algorithm~\cite{BermanBMRY13}. We explain our algorithm in the language of the shortcut problem. 

\subsection{An overview}
In this section, we give an overview of \Cref{thm:upperbound}. Our goal is to find an \ssss{\apxS s}{\apxD d} in a given graph $G$ and integers $d$ and $s$ such that $G$ admits an \ssss{s}{d}. Our algorithm incorporates many ideas from the previous algorithm of \cite{BermanBMRY13}. Below, we present an overview and discuss how our techniques differ from theirs.

Denote by $\pset \subseteq V(G) \times V(G)$ the (ordered) set of reachable pairs in $G$. 
The shortcut edges are chosen from $E^T \setminus E$ where $E^T$ are the edges in the transitive closure of $G$. We say that $F \subseteq (E^T \setminus E)$ $d'$-settles the pair $(u,v)$ if the distance between them in $E \cup F$ is at most $d'$.  
Therefore, to get $(\alpha_D, \alpha_S)$-approximation, it suffices to compute a subset $F \subseteq E^T \setminus E$ so that every pair in $\pset$ is $(\alpha_D \cdot d)$-settled by $F$.  
We will handle two types of pairs in $\pset$ separately. 
We say that  $(u,v) \in \pset$ is \textit{thick} if the total number of vertices $w$ that are reachable from $u$ and can reach $v$ is at least $\beta$; otherwise, the pair $(u,v)$ is \textit{thin}. Intuitively, a thick pair is ``very well connected''. 
Divide $\pset$ into $\pset = \pset_{thick} \cup \pset_{thin}$.

The high-level ideas in dealing with these cases roughly follow~\cite{BermanBMRY13}. At a high-level, a generic approach to turn a (single-criteria) approximation algorithm into a bicriteria one is to prove a certain ``scaling advantage'' result, e.g., proving that a $(1,\alpha_S)$-approximation implies $(\alpha_D, O(\alpha_S/\alpha_D))$-approximation (so when $\alpha_D=1$, we achieve roughly the same result).

\paragraph{Settling thick pairs.} Sample $V' \subseteq V(G): |V'| =  \tilde{\Theta}(n/\beta)$ uniformly. For each sampled vertex $v \in V'$, add edges from $v$ to all $v'$ such that $v'$ is reachable from $v$ or can reach $v$. 
Denote by $F_1$ the set of edges that are added by this process, so $|F_1| = \tilde{\Theta}(n^2/\beta)$. It is easy to see (via a hitting set argument) that, with constant probability, this set of shortcut edges reduces the diameter to two and the total number of shortcut edges is at most $\tilde O(n^2/\beta)$.

This simple sampling strategy has already been explored by \cite{BermanBMRY13} in the context of spanner construction.\footnote{In the context of spanners, instead of connecting a sampled $v$ with in- and out-edges with vertices in $V$, they consider in- and out-arborescence rooted at $v$. In our context (of TC-spanners and shortcuts), these arborescences are exactly the edges we described above.} 
To incorporate the scaling advantage, we adapt the technique used in a recent paper demonstrating how to construct an \ssss{n}{\tOh(n^{1/3})} for any graph~\cite{KoganP22}. This adaptation reduces the size of the shortcut set to $\tilde{O}\left(n^2/\beta(\apxD d)^2\right)$ while still ensuring that the endpoints of all thick edges   have a distance of at most $\apxD d$.

\paragraph{Settling the thin pairs.} Thin pairs are handled via LP-rounding techniques. 
We say that a set $A\subseteq E^T \backslash E$ is $d'$-\textit{critical} for thin pair $(u,v) \in \pset_{thin}$ if in $E^T \setminus A$, the distance from $u$ to $v$ is larger than $d$; in other words, not taking any edge from $A$ would make the solution infeasible. 
Denote by $\aset_{d'}$ the set of all minimal and $d'$-critical sets.  
Our definition of critical set is analogous to the notion of antispanners in~\cite{BermanBMRY13}. The following claim is intuitive: It asserts an alternative characterization of shortcuts as a set of edges that hit every critical set.

\begin{claim}[Adapted from \cite{BermanBMRY13}] \label{clm:shortcut-hittingset}
    An edge set $E'$ is a $d'$-shortcut set for all thin pairs if and only if $E' \cap A \neq \emptyset$  for all $A\in\mathcal{A}_{d'}$. 
\end{claim}

The above claim allows us to write the following LP constraints for finding a $(d,s)$-shortcut.

\begin{align*}
	\text{(LP)}: \qquad\sum_{e\in E^T\backslash E}x_e &\le s \\
	\sum_{e\in A}x_e &\ge 1\qquad \forall A\in \mathcal{A}_d\\
	x_e &\geq 0\qquad \forall e\in E^T \setminus E
\end{align*}

Each variable $x_e$ indicate whether edge $e \in E^T \setminus E$ is included into the shortcut solution. Each constraint $\sum_{e \in A} x_e \geq 1$ asserts that critical edge set $A$ must be ``hit'' by the solution. There can be exponentially many constraints, but an efficient separation oracle exists, as we sketch below.

The high-level idea of~\cite{BermanBMRY13} (when translated into our setting) is a ``round-or-cut'' procedure that, from a feasible solution $\textbf{x} \in [0,1]^{E^T \setminus E}$, randomly computes $F_2 \subseteq E^T \setminus E$ of  size $O(\beta s)$ that either (i) successfully $d$-settles all thin pairs in $\pset_{thin}$ or (ii) can be used to find a critical set $A \in \aset_{d}$ such that $\sum_{e \in A} x_e < 1$. 
If Case (i) happens, we have successfully computed the solution $F_2$. Otherwise, when Case (ii) happens, we have a separation oracle. 
In order to incorporate the scaling advantage into this algorithm, we prove a decomposition lemma (\cref{lem:findantispanner}) which asserts that each $(\alpha_D \cdot d)$-critical set $A \in\mathcal{A}_{\apxD d}$ can be partitioned into $\apxD$  sets in $\mathcal{A}_{d}$, and such decomposition can be computed efficiently.
This decomposition lemma allows us to generalize~\cite{BermanBMRY13} to the bicriteria setting, leading to the set of size $O(\beta s/\alpha_D)$ that $(\alpha_D d)$-settles thin pairs.

\subsection{Preliminaries}

\paragraph{Large diameter case:} If we aim for a relatively large diameter, i.e., $d \cdot \alpha_D \geq n^{0.34}$, we can simply use the known result. 

\begin{theorem}[\cite{KoganP22}]
There is an efficient algorithm that, given input graph $G$, computes a $(d,s)$-shortcut for $d = \tilde{O}(n^{1/3})$ and $s = \tilde{O}(n)$.    
\end{theorem}

Using the above theorem, if $\apxD{}d=\Omega(n^{0.34})$, we can easily achieve $(\alpha_D, \alpha_S)$-approximation for all $\alpha_S = \tilde{O}(1)$. 
Therefore, we make the following assumption throughout this section. 

\begin{remark}\label{rem:assumption}
We can assume w.l.o.g. that $\alpha_D d = O(n^{0.34})$. 
\end{remark}

\paragraph{Reduction to DAGs:} We argue that DAGs, in some sense, capture hard instances for our problems. This will allow us to focus on DAGs in the subsequent sections.
The formal statement is encapsulated in the following lemma. 

\begin{lemma}
If there exists an efficient $(\alpha_D,\alpha_S)$ approximation algorithm for DAGs, then there exists an efficient $(3\alpha_D, 3\alpha_S)$-approximation algorithm for all directed graphs.  
\end{lemma}
\begin{proof}
Assume that we are given access to the algorithm $\aset(G,s,d)$ that produces $(\alpha_D,\alpha_S)$ approximation algorithm for the DAG case. 
Let $G$ be an input digraph, together with the input parameters $(d,s)$.  
Compute a collection ${\mathcal S}$ of strongly connected components (SCC) of $G$, and let $G'$ be the DAG obtained by contracting each SCC into a single node. Invoke the algorithm $\aset(G',s,d)$. Let $E' \subseteq E(G')$ be the shortcut edges so that $|E'| \leq \alpha_S \cdot s$ and the diameter of $G' \cup E'$ is at most $\alpha_D \cdot d$. 
These edges would be responsible for connecting the pairs whose endpoints are in distinct SCCs. 

For each SCC $C \in {\mathcal S}$, we pick an arbitrary ``center'' $u_C \in V(C)$ and connect it to every other vertex in $V(C)$. These edges are called $E''_C$. Define $E'' = \bigcup_{C \in {\mathcal S}} E''_C$. 
The final shortcut $F$ can be constructed by combining these two sets $E'$ and $E''$: Edges in $E''$ can be added into $F$ directly. For each edge that connects $C$ to $C'$ in $E'$, we create the corresponding edge $(u_C, u_C')$ connecting the centers. Notice that $|F| \leq |E'| + |E''| \leq \alpha_S \cdot s + 2n \leq 3\alpha_S \cdot s$ (here we used the assumption that $s \geq n$). Moreover, it is easy to verify that, for each reachable pair $(v,w)$ in $G$ where $v \in C$ and $w \in C'$, there is a path from $v$ to $w$ of length at most $3 \alpha_D \cdot d$.     
\end{proof}

\subsection{Framework} \label{sec:ub-overview}

Suppose $G=(V,E)$ is a directed acyclic graph and $G^T=(V,E^T)$ be its transitive closure, i.e., $(u,v)\in E^T$ if $u$ has a directed path with length at least $1$ to $v$ in $G$. Since $G$ is acyclic, $G^T$ contains no self loop. 

\begin{definition}[Local graphs]\label{def:localgraph}
For a pair $u,v \in V(G)$, we define the local graph $G^{u,v}$. 
 Let $G^{u,v}=(V^{u,v},E^{u,v})$ be the subgraph of $G^T$ induced by the vertices that can reach $v$ and can be reached from $u$ (i.e., these vertices lie on at least one path from $u$ to $v$). 
\end{definition}
\begin{definition}[Thick and thin pairs]\label{def:thickthinedge}
 Let $u,v \in V(G)$. 
 If $|V^{u,v}|\ge \beta$ ($1\le \beta\le n$ will be determined later), the corresponding edge $(s,t)$ is said to be $\beta$-thick, and otherwise, it is $\beta$-thin. When $\beta$ is clear from context, we will simply write thick and thin pairs respectively. 
\end{definition}

Denote by $\pset$ the set of pairs of vertices that are reachable in $G$. We can partition $\pset$ into $\pset_{thick} \cup \pset_{thin}$. 

\begin{definition}
	Let $d' \in {\mathbb N}$. A set $E'\subseteq E^T \setminus E$ is said to $d'$-settle a pair $(u,v)\in E$ if $(V,E \cup E')$ contains a path of length at most $d'$ from $u$ to $v$. 
\end{definition}

Our algorithm will find two edge sets $F_1,F_2\subseteq E^T \setminus E$ such that the set $F_1$ is responsbile for $(\apxD{} d)$-settling all thick pairs, while $F_2$ will $(\apxD{} d)$-settle all thin pairs. The final solution be $F_1\cup F_2$ (Notice that $F_1\cup F_2$ $(\apxD{} d)$-settles all the pairs.)  These are encapsulated in the following two lemmas: 

\begin{restatable}[thick pairs]{lemma}{settlethick}\label{lem:settlethickedges}
We can efficiently compute $F_1$ such that $|F_1| \le O\left(\frac{n^2\log^2 n}{\beta\apxD{}^2 d^2}+
n \log n
\right)$ and it $(\apxD{} d)$-settles all thick pairs w.h.p. 
\end{restatable}

\begin{lemma}[thin pairs] \label{lem:settlethinpairs}
We can efficiently compute $F_2$ such that $|F_2| \leq O\left(\frac{\beta \log^2 n s}{\alpha_D}\right)$ and $(\alpha_D d)$-settles all thin pairs with high probability.      
\end{lemma}

We will prove these two lemmas later. Meanwhile, we complete the proof of Theorem~\ref{thm:upperbound}. 
We minimize the sum $|F_1| + |F_2|$ by setting the value of $\beta=\frac{n}{d\sqrt{s\apxD{}}}$.\footnote{The only problem is that this value could be much less than $1$ when $d\sqrt{s\apxD{}}>n$, which leads to $(d\apxD{})^2s>n^2$. However, in that case,  we can use the known tradeoff~\cite{KoganP22} to construct a \ssss{s}{\apxD{}d} when $d\sqrt{s\apxD{}}>n$.} 
	Now we have 
 \[|F_1\cup F_2|=O\left(\frac{ns^{0.5}\log^{2}n}{d\apxD{}^{1.5}}+n\log n\right)\le s\cdot O\left(\frac{n\log^{2}n}{d\apxD{}^{1.5}s^{0.5}}\right).\]

\subsection{Settling the thick pairs}\label{subsec:thickedges}

In this Section, we prove Lemma~\ref{lem:settlethickedges}.  For convenience we assume $\apxD{}d=\omega(\log n)$. We will show the case when $\apxD{}d=O(\log n)$ later.
We will use the idea from~\cite{KoganP22} to construct $F_1$. First, the following lemma allows us to decompose a DAG into a collection of paths and independent sets. 

\begin{lemma}[Theorem 3.2 \cite{GrandoniILPU21}]
	There is a polynomial time algorithm given an $n$-vertex acyclic graph $G=(V,E)$ and an integer $k \in[1,n]$, partition $G$ into $k$ directed paths  $P_1,...,P_{k}$ and at most $2n/k$ independent set  $Q_1,...,Q_{2n/k}$ in $G^T$. In other words, $P_1,...,P_{k},Q_1,...,Q_{2n/k}$ are disjoint and $\left(\cup_{i\in[k]}P_i\right)\cup \left(\cup_{i\in[2n/k]}Q_i\right)=V$.
\end{lemma}

We first apply the following lemma with $k=n/(\apxD{} d)$ to get $P_1,...,P_{8n/(\apxD{} d)}$ and $Q_1,...,Q_{\apxD{} d/4}$. Notice that $\apxD d=\omega(\log n)$ and $\apxD d= O(n^{0.34})$, we can safely assume $8n/(\apxD{} d)$ and $\apxD{} d/4$ are integers without loss of generality.
We will use Lemma 1.1~\cite{Raskhodnikova10}.
\begin{lemma}[Lemma 1.1~\cite{Raskhodnikova10}]\label{lem:pathshortcut}
    For any integer $n\ge 3$, the directed path with length $\ell$ has a $2$-shortcut with at most $\ell\log \ell$ edges.
\end{lemma}

We first add $n\log n$ edges to $F_1$ and reduce the diameter of every path $P_i$ to $2$. To accomplish this, for each path, we use \cref{lem:pathshortcut}. Since different paths are disjoint regarding vertices, $n\log n$ edges suffice.

Next, let $R \subseteq V$ be obtained by sampling $\min\left((9\log n)\cdot n/\beta, n\right)$ vertices from $V$ uniformly at random. The algorithm also samples $\min\left((999\log n)\cdot n/(\apxD{} d)^2,n/(\apxD{} d)\right)$ paths from $P_1,...,$ $P_{n/(\apxD{} d)}$ uniformly at random; let ${\mathcal Q}$ denote the set of sampled paths. For each vertex $u\in R$ and path $p\in {\mathcal Q}$, add $(u,v_1)$ to $F_1$ where $v_1$ is the first node in $p$ that $u$ can reach (if it exists), and add $(v_2,u)$ to $F_1$ where $v_2$ is the last node in $p$ that can reach $u$ (if it exists). Observe that $\beta$ is between $1$ and $n$, $(\apxD d)$ is between $\omega(\log n)$ and $O(n^{0.34})$, so both $(9\log n)\cdot n/\beta$ and $(999\log n)\cdot n/(\apxD{} d)^2$ will not be too small, and we can safely assume they are integers without loss of generality.

\settlethick*
\begin{proof}
It is straightforward to verify that $|F_1| \le (999n^2\log^2 n)/(\beta\apxD{}^2 d^2)+n\log n$ by the construction.

Suppose $(s,t)\in E^T$ is a thick pair. Since $|V^{s,t}|\ge \beta$, a vertex $u\in R\cap V^{s,t}$ exists with high probability (w.h.p.). %
Let $G'$ be the graph obtained after adding all edges that reduce the diameter for each $P_i$. Denote the shortest path between $s$ and $t$ in $G'$ as $P_{s,t}$. 
The path $P_{s,t}$ intersects each path $P_i$ at no more than three nodes. Otherwise, there are four nodes in $P_i\cap P_{s,t}$ $v_1,v_2,v_3,v_4$ such that $v_1$ has distance at least $3$ to $v_4$ since $P_{s,t}$ is a shortest path, which contradict the fact that $P_i$ has diameter $2$. The path $P_{s,t}$ intersects each $Q_i$ at no more than one node since $Q_i$ is an independent set on $E^T$.
Therefore, if we disregard all nodes in $Q_i$ from $i=1$to $i=\apxD{} d/ 4$ (which incurs an additional $\apxD{} d/4$ steps) and examine the first $\apxD{} d/4$ vertices and the last $\apxD{} d/4$ vertices of $P_{s,t}$, both of them will intersect a path in ${\mathcal Q}$ w.h.p. since we sample $(999\log n)\cdot n/(\apxD{} d)^2$ paths into ${\mathcal Q}$ among all $8n/(\apxD d)$ paths. This implies that $s$ can first use a $\apxD{} d/4+\apxD{} d/9+1$ path to reach $u$, and then $u$ can use another $1+\apxD{} d/9+\apxD{} d/4$ path to reach $t$. 
\end{proof}

For the case when $\apxD{}d=O(\log n),\apxD{}d>1$, to get similar result as~\cref{lem:settlethickedges}, we want $|F_1|=O(n^2\log n/\beta+n)$. This is easy to construct by sampling $(n/\beta)\log n$ nodes, where w.h.p. one of them will be in $V^{s,t}$ for every thick pair $s,t$. For every sample node $u$, we just need to include $(v,u)$ to $F_1$ for every $v$ that can reach $u$, and $(u,v)$ to $F_1$ for every $v$ that $u$ can reach. Now every thick pair has distance $2$ in $F_1$. 

For the extreme case when $\apxD{}d=1$, there is a unique way of adding edges, which is connecting every reachable pair by one edge, thus, we ignore this situation.

\subsection{Settling the thin pairs}\label{subsec:thinedges}

\begin{definition}[Critical sets]\label{def:antispanner}
A set $A\subseteq E^T$ is a $k$-critical set of a pair $(u,v)\in E^T$ if $E^T\backslash A$ contains no path from $u$ to $v$ with a length of at most $k$. If there does not exist an $A'\subset A$ such that $A'$ is also a $k$-critical, then we say $A$ is minimal.
\end{definition}

\begin{definition}[$\mathcal{A}_k$]\label{def:antispannerset}
Let $\mathcal{A}_k$ consist of all sets $A$ satisfying both (i) $A$ is a minimal critical set of some thin pair, and (ii) $A\cap E=\emptyset$.
\end{definition}

This gives an alternate characterization of a shortcut set. 

\begin{claim}[Adapted from \cite{BermanBMRY13}] 
    An edge set $E'$ is a $d$-shortcut set for all thin pairs if and only if $E' \cap A \neq \emptyset$  for all $A\in\mathcal{A}_{d}$. 
\end{claim}

\begin{proof}
We briefly sketch the proof of this claim below.
\begin{itemize}
    \item[($\Rightarrow$)] We first argue that $E'$ must intersect each $A\in\mathcal{A}_d$ with at least one edge. Suppose, on the contrary, that $E'$ is disjoint from a set $A\in\mathcal{A}_d$, which is a $d$-critical set of some pair $(u,v)$. Then, $E'\cup E\subseteq E^T \backslash A$, and the distance between $u$ and $v$ in $E'\cup E$ is larger than $d$, which contradicts our assumption. 
    
    \item[($\Leftarrow$)] On the other hand, assume that $E'$ intersects every $A\in\mathcal{A}_d$ with at least one edge then it must be a $d$-shortcut. Otherwise, there exists a pair $(u,v)\in E^T $ such that, in $E'\cup E$, $u$ has a distance larger than $d$ to $v$. This implies that $A= E^T\backslash(E'\cup E)$ is a $d$-critical set of $u,v$ and $A \cap E' = \emptyset$ (a contradiction). 
\end{itemize}
\end{proof}

We define the following polytope $\cP_{G,d}$. It includes at most $n^2$ variables: for each $e\in E^T$, there is a corresponding variable $x_e$. The polytope has an exponentially large number of constraints. However, it must be non-empty since we assume $G$ admits a \ssss{s}{d}.

\begin{align*}
	\text{Polytope }\cP_{G,d}: \qquad\sum_{e\in E^T\backslash E}x_e &\le s \\
	\sum_{e\in A}x_e &\ge 1\qquad \forall A\in \mathcal{A}_d\\
	x_e &\geq 0\qquad \forall e\in E^T
\end{align*}

The algorithm employs the cutting-plane method to attempt to find a point within the polytope. According to established analyses, the running time of the cutting-plane method is polynomial with respect to the number of variables, provided that a separation oracle with a polynomial running time exists based on the number of variables. A separation oracle either asserts that the point lies within the polytope or returns a violated constraint, which can function as a cutting plane. Since the constraints $\sum_{e\in E^T\backslash E}x_e\le s$ and $x_e\ge 0$ can be verified in polynomial time, identifying a cutting plane is straightforward if either of these constraints is violated. Consequently, we assume $\sum_{e\in E^T\backslash E}x_e\le s$ and $x_e\ge 0$.

In the subsequent algorithm, we accept a point as input and either return a non-trivial violated constraint (one of $\sum_{e\in A}x_e\ge 1$), which can act as a separation plane, or output a set $E_2$ that settles all thin pairs.

\begin{algorithm}[H]
	
	\caption{{\sc Cut-or-Round}$({\bf x})$}\label{alg:seperationoracle}
	
	\KwData{A vector ${\bf x} \in [0,1]^{E^T \setminus E}$ satisfying $\sum_{e\in E^T\backslash E}x_e\le s$ and $x_e\ge 0$.}
	\KwResult{A set $A \in  \mathcal{A}_d$ s.t. $\sum_{e\in A}x_e<1$, or a set $F_2\subset E^T \setminus E$ that $(\apxD{} d)$-settles all thin pairs.}
	Include edge $e\in E^T$ independently into $F_2$ with probability $(500\log n)(\beta/\apxD{})x_e$\;
	
	\If{all thin pairs are $(\apxD{} d)$-settled}{\If{$|F_2\backslash E|\le (1000\log^2 n)(\beta/\apxD{})s$}{\Return{$F_2$}\;}\Else{\Return{fail}\;}}
	\Else{
		Use $F_2$ to find a $(\apxD{} d)$-critical set $A'\in\mathcal{A}_{\apxD{} d}$ with $F_2\cap A'=\emptyset$ using the algorithm from Claim 2.4 in~\cite{BermanBMRY13}\;
		\If{$\sum_{e\in A'}x_e<\apxD{}/9$}{
			Find a $d$-critical set $A\subseteq \mathcal{A}_d$  such that $\sum_{e\in A}x_e\ge 1$ using~\cref{lem:findantispanner}\;
			\Return violated constraint $\sum_{e\in A}x_e\ge 1$\;}
		\Else{\Return{fail}\;}
	}
\end{algorithm}

\begin{claim}\label{lem:notfail}
	The algorithm will fail with probability at most $\frac{1}{n^{\omega(1)}}$.
\end{claim}
\begin{proof}
	The first fail condition is $|F_2\backslash E|>(1000\log^2 n)(\beta/\apxD{})s$. Notice that each edge $e$ is included in $F_2$ independently with probability $(500\log n)(\beta/\apxD{})x_e$ where $\sum_{e\in E^T\backslash E}x_e\le s$. By Chernoff bound, $|F_2\backslash E|>(1000\log^2 n)(\beta/\apxD{}) s$ happens with probability at most $e^{-(\log n)(\beta/\apxD{}) s}$. Remember that $s\ge n$ and $\apxD{}=O(n^{0.34})$, which means we have $e^{-(\log n)(\beta/\apxD{}) s}\le n^{-\omega(1)}$.
	
	Now we show that the second fail condition happens with small probability. We define event $\mathcal{E}$ as ``all $B\subseteq A_{\apxD{} d}$ with $\sum_{e\in B}x_e\ge \apxD{}/9$ satisfies $B\cap F_2\not=\emptyset$''. If the second fail is triggered, then there exists a set $A'\in\mathcal{A}_{\apxD{} d}$ with $\sum_{e\in A'}x_e\ge \apxD{}/9$ satisfies $A'\cap F_2=\emptyset$, which means $\mathcal{E}$ does not happen. Thus, the probability that the second fail is triggered is bounded by $1-\Pr[\mathcal{E}]$. For $B\subseteq A_{\apxD{} d}$ with $\sum_{e\in B}x_e\ge \apxD{}/9$, the probability that $B\cap F_2=\emptyset$ is bounded by $\exp\left(-9\beta\log n\right)$ according to Chernoff bound (remember that each edge is included in $F_2$ with probability $(500\log n)(\beta/\apxD{})x_e$). Now we count the size of $\mathcal{A}_{\apxD{} d}$. According to Claim 2.5 in~\cite{BermanBMRY13}, we have $|\mathcal{A}_{\apxD{} d}|\le |E|\cdot \beta^\beta\le \exp\left(2\beta\log\beta\right)$. Finally, by using union bound, $\Pr[\mathcal{E}]\ge 1-\frac{1}{n^{\omega(1)}}$.
\end{proof}
\begin{lemma}[critical set decomposition]\label{lem:findantispanner}
	There exists a polynomial time algorithm that, given $A'\in\mathcal{A}_{\apxD{} d}$ with $\sum_{e\in A'}x_e< \apxD{}/9$, outputs $A\in\mathcal{A}_{d}$ such that $\sum_{e\in A}x_e< 1$. 
\end{lemma}
\begin{proof}
	The algorithm first use polynomial time to find $(s,t)\in E^T$ such that $A'$ is a $(\apxD{} d)$-critical set of $(s,t)$. Then the algorithm constructs a shortest path tree rooted at $s$ on the subgraph $E^T\backslash A'$, where all the nodes in the $i$-th layer of the tree has distance $i$ from $s$. Denote the vertex set of the $i$-th layer as $L_i$. According to the definition of critical set, $t$ is on at least the $(\apxD{} d+1)$-th layer. The algorithm devides the first $\apxD{} d$ layers into at least $\apxD{}/3$ batches, where the $i$-th batch contains all vertices between layer $2(i-1)d$ to $2i\cdot d-1$. Denote the set of all the edges in $A'$ that has at least one endpoint in the $i$-th batch as $A_i$. Since each edge in $A'$ will be included in at most two $A_i$, it is easy to see that at least one of $A_i$ has the property $\sum_{e\in A_i}x_e<1$. 
	
	Now we prove that for any $i$, $A_i\in\mathcal{A}_d$. Since $A_i$ is a subset of $A'$, the second condition, i.e., $A_i\cap E=\emptyset$ is satisfied trivially. We only need to verify that $A_i$ is a $d$-critical set of some edge in $E^T$. The idea is to choose a vertex in $L_{2(i-1)d}$ and another vertex in $L_{2id-1}$, and argue that they are reachable and have distance more than $d$ in $E^T\backslash A_i$. Let $S$ contain all the vertices in $L_{2(i-1)d}$ that has at least one edge towards a vertex in $L_{2id-1}$. Notice that $G$ (and also $G^T$) is acyclic (\cref{rem:assumption}), which also means the induced subgraph $G^T[S]$ is acyclic. Thus, there exists a vertex $u\in S$ such that it has no edge in $E^T$ towards other vertices in $S$. Take an arbitrary vertex $v\in L_{2id-1}$ such that $(u,v)\in E^T$, now we argue that $u$ has distance more than $d$ to $v$ in $E^T\backslash A_i$.  
	
	We prove it by induction. Induction hypothesis: any vertex that $u$ has distance $x$ to in $E^T\backslash A_i$ is one of the following two types (i) a vertex that cannot reach $v$ in $G$ (ii) a vertex in $L_{2(i-1)d+x}$. If the induction hypothesis is correct for any $0\le x\le d+10$, then $u$ cannot have distance at most $d$ to $v$ in $E^T\backslash A_i$. Now we prove the induction hypothesis. When $x=0$, the hypothesis is correct. For $x>0$, suppose $w_2$ is a vertex with $\distt{E^T\backslash A_i}{u,w_2}=x$, then there exists an edge $(w_1,w_2)\in E^T\backslash A_i$ such that $\distt{E^T\backslash A_i}{u,w_1}=x-1$. According to the induction hypothesis, we can assume $w_1$ is either type (i) or type (ii). If $w_1$ is type (i), then $w_2$ cannot reach $v$ in $G$, which means $w_2$ is also type (i). Now suppose $w_1$ is type (ii) and can reach $v$ in $G$. First of all, $w_2$ cannot be in the first $2(i-1)d$-th layer, otherwise $w_2$ has a path $p$ in $G$ to $v$ where $p$ must contain a vertex in $L_{2(i-1)d}$ (two consecutive vertices in $p$ cannot skip a layer since all edges in $p$ is in $E$, which is also in $E^T\backslash A'$), which means $u$ has an edge in $E^T$ to another vertex in $L_{2(i-1)d}$, leading to a contradiction. Then, $w_2$ cannot be a vertex in $L_{2(i-1)d+b}$ where $1\le b<x$ since $u$ has distance less than $x$ to them according to induction hypothesis. $w_2$ cannot be a vertex in $L_{2(i-1)d+b}$ where $b>x$, otherwise $(w_1,w_2)$ is not in $A_i$, also not in $A'$, which contradicts the fact that layers are constructed by shortest path tree on $E^T\backslash A'$. $w_2$ cannot be a vertex outside the tree because $s$ can reach $w_2$ in $E^T\backslash A'$. Finially, $w_2\in L_{2(i-1)d+x}$ and the induction hypothesis holds. 
\end{proof}

As mentioned in \Cref{sec:ub-overview}, Theorem~\ref{thm:upperbound} now follows combining the size of sets $F_1$ and $F_2$ from \Cref{lem:settlethickedges} and \Cref{lem:settlethinpairs} respectively.

%% file: main.bbl
\newcommand{\etalchar}[1]{$^{#1}$}
\begin{thebibliography}{BBM{\etalchar{+}}13}

\bibitem[AGU72]{AhoGU72}
Alfred~V. Aho, M.~R. Garey, and Jeffrey~D. Ullman.
\newblock The transitive reduction of a directed graph.
\newblock {\em {SIAM} J. Comput.}, 1(2):131--137, 1972.

\bibitem[BBM{\etalchar{+}}11]{BermanBMRY11}
Piotr Berman, Arnab Bhattacharyya, Konstantin Makarychev, Sofya Raskhodnikova,
  and Grigory Yaroslavtsev.
\newblock Improved approximation for the directed spanner problem.
\newblock In {\em {ICALP} {(1)}}, volume 6755 of {\em Lecture Notes in Computer
  Science}, pages 1--12. Springer, 2011.

\bibitem[BBM{\etalchar{+}}13]{BermanBMRY13}
Piotr Berman, Arnab Bhattacharyya, Konstantin Makarychev, Sofya Raskhodnikova,
  and Grigory Yaroslavtsev.
\newblock Approximation algorithms for spanner problems and directed steiner
  forest.
\newblock {\em Inf. Comput.}, 222:93--107, 2013.

\bibitem[BGJ{\etalchar{+}}12]{BhattacharyyaGJRW12}
Arnab Bhattacharyya, Elena Grigorescu, Kyomin Jung, Sofya Raskhodnikova, and
  David~P. Woodruff.
\newblock Transitive-closure spanners.
\newblock {\em {SIAM} J. Comput.}, 41(6):1380--1425, 2012.
\newblock announced at SODA'11.

\bibitem[BH23]{BodwinH23}
Greg Bodwin and Gary Hoppenworth.
\newblock Folklore sampling is optimal for exact hopsets: Confirming the
  {\(\surd\)}n barrier.
\newblock In {\em {FOCS}}, pages 701--720. {IEEE}, 2023.

\bibitem[CFR20]{CaoFR20}
Nairen Cao, Jeremy~T. Fineman, and Katina Russell.
\newblock Improved work span tradeoff for single source reachability and
  approximate shortest paths.
\newblock In {\em {SPAA}}, pages 511--513. {ACM}, 2020.

\bibitem[CK09]{chuzhoy2009polynomial}
Julia Chuzhoy and Sanjeev Khanna.
\newblock Polynomial flow-cut gaps and hardness of directed cut problems.
\newblock {\em Journal of the ACM (JACM)}, 56(2):1--28, 2009.

\bibitem[CLN13]{chalermsook2013graph}
Parinya Chalermsook, Bundit Laekhanukit, and Danupon Nanongkai.
\newblock Graph products revisited: Tight approximation hardness of induced
  matching, poset dimension and more.
\newblock In {\em Proceedings of the twenty-fourth annual ACM-SIAM symposium on
  Discrete algorithms}, pages 1557--1576. SIAM, 2013.

\bibitem[CLN14]{chalermsook2014pre}
Parinya Chalermsook, Bundit Laekhanukit, and Danupon Nanongkai.
\newblock Pre-reduction graph products: Hardnesses of properly learning dfas
  and approximating edp on dags.
\newblock In {\em 2014 IEEE 55th Annual Symposium on Foundations of Computer
  Science}, pages 444--453. IEEE, 2014.

\bibitem[DKN25]{DinitzKN25}
Michael Dinitz, Ama Koranteng, and Yasamin Nazari.
\newblock Approximation algorithms for optimal hopsets.
\newblock In {\em {ICALP}}, volume 334 of {\em LIPIcs}, pages 69:1--69:20.
  Schloss Dagstuhl - Leibniz-Zentrum f{\"{u}}r Informatik, 2025.

\bibitem[DKR16]{DinitzKR16}
Michael Dinitz, Guy Kortsarz, and Ran Raz.
\newblock Label cover instances with large girth and the hardness of
  approximating basic \emph{k}-spanner.
\newblock {\em {ACM} Trans. Algorithms}, 12(2):25:1--25:16, 2016.

\bibitem[EP07]{ElkinP07}
Michael Elkin and David Peleg.
\newblock The hardness of approximating spanner problems.
\newblock {\em Theory Comput. Syst.}, 41(4):691--729, 2007.

\bibitem[Fin18]{Fineman18}
Jeremy~T. Fineman.
\newblock Nearly work-efficient parallel algorithm for digraph reachability.
\newblock In Ilias Diakonikolas, David Kempe, and Monika Henzinger, editors,
  {\em Proceedings of the 50th Annual {ACM} {SIGACT} Symposium on Theory of
  Computing, {STOC} 2018, Los Angeles, CA, USA, June 25-29, 2018}, pages
  457--470. {ACM}, 2018.

\bibitem[GIL{\etalchar{+}}21]{GrandoniILPU21}
Fabrizio Grandoni, Giuseppe~F. Italiano, Aleksander Lukasiewicz, Nikos
  Parotsidis, and Przemyslaw Uznanski.
\newblock All-pairs {LCA} in dags: Breaking through the
  \emph{O}(\emph{n}\({}^{\mbox{2.5}}\)) barrier.
\newblock In {\em {SODA}}, pages 273--289. {SIAM}, 2021.

\bibitem[GKR{\etalchar{+}}99]{guruswami1999near}
Venkatesan Guruswami, Sanjeev Khanna, Rajmohan Rajaraman, Bruce Shepherd, and
  Mihalis Yannakakis.
\newblock Near-optimal hardness results and approximation algorithms for
  edge-disjoint paths and related problems.
\newblock In {\em Proceedings of the thirty-first Annual ACM Symposium on
  Theory of Computing}, pages 19--28, 1999.

\bibitem[Hes03]{Hesse03}
William Hesse.
\newblock Directed graphs requiring large numbers of shortcuts.
\newblock In {\em {SODA}}, pages 665--669. {ACM/SIAM}, 2003.

\bibitem[HP21]{HuangP21}
Shang{-}En Huang and Seth Pettie.
\newblock Lower bounds on sparse spanners, emulators, and diameter-reducing
  shortcuts.
\newblock {\em {SIAM} J. Discret. Math.}, 35(3):2129--2144, 2021.

\bibitem[JST20]{JinST20}
Yujia Jin, Aaron Sidford, and Kevin Tian.
\newblock Semi-streaming bipartite matching in fewer passes and less space.
\newblock {\em CoRR}, abs/2011.03495, 2020.

\bibitem[Kor01]{Kortsarz01}
Guy Kortsarz.
\newblock On the hardness of approximating spanners.
\newblock {\em Algorithmica}, 30(3):432--450, 2001.

\bibitem[KP22a]{KoganP-icalp22}
Shimon Kogan and Merav Parter.
\newblock Beating matrix multiplication for n{\^{}}\{1/3\}-directed shortcuts.
\newblock In {\em {ICALP}}, volume 229 of {\em LIPIcs}, pages 82:1--82:20.
  Schloss Dagstuhl - Leibniz-Zentrum f{\"{u}}r Informatik, 2022.

\bibitem[KP22b]{KoganP22}
Shimon Kogan and Merav Parter.
\newblock New diameter-reducing shortcuts and directed hopsets: Breaking the
  barrier.
\newblock In {\em {SODA}}, pages 1326--1341. {SIAM}, 2022.

\bibitem[KP23]{KoganP23}
Shimon Kogan and Merav Parter.
\newblock Faster and unified algorithms for diameter reducing shortcuts and
  minimum chain covers.
\newblock In {\em {SODA}}, pages 212--239. {SIAM}, 2023.

\bibitem[Mos15]{Moshkovitz15}
Dana Moshkovitz.
\newblock The projection games conjecture and the np-hardness of ln
  n-approximating set-cover.
\newblock {\em Theory Comput.}, 11:221--235, 2015.

\bibitem[Ras10a]{Raskhodnikova10}
Sofya Raskhodnikova.
\newblock Transitive-closure spanners: {A} survey.
\newblock In {\em Property Testing}, volume 6390 of {\em Lecture Notes in
  Computer Science}, pages 167--196. Springer, 2010.

\bibitem[Ras10b]{Ras10}
Sofya Raskhodnikova.
\newblock Transitive-closure spanners: {A} survey.
\newblock In {\em Property Testing - Current Research and Surveys}, volume 6390
  of {\em Lecture Notes in Computer Science}, pages 167--196. Springer, 2010.

\bibitem[Tho92]{Thorup92}
Mikkel Thorup.
\newblock On shortcutting digraphs.
\newblock In Ernst~W. Mayr, editor, {\em Graph-Theoretic Concepts in Computer
  Science, 18th International Workshop, {WG} '92, Wiesbaden-Naurod, Germany,
  June 19-20, 1992, Proceedings}, volume 657 of {\em Lecture Notes in Computer
  Science}, pages 205--211. Springer, 1992.

\bibitem[Tho95]{Thorup95}
Mikkel Thorup.
\newblock Shortcutting planar digraphs.
\newblock {\em Comb. Probab. Comput.}, 4:287--315, 1995.

\bibitem[VXX24]{WilliamsXX24}
Virginia {Vassilevska Williams}, Yinzhan Xu, and Zixuan Xu.
\newblock Simpler and higher lower bounds for shortcut sets.
\newblock In {\em {SODA}}, pages 2643--2656. {SIAM}, 2024.

\end{thebibliography}
